\documentclass[runningheads]{llncs}

\input{preamble}

\newcommand\defaccr[2]{\newcommand#1{#2\xspace}}
\newcommand\defmath[2]{\newcommand#1{\ensuremath{#2}\xspace}}
\newcommand\concept[1]{\textit{#1}}

\defmath\CG{\mathcal C_n}

\defaccr\ERROR{Wrong} 
\defaccr\timeout{$>300$}
\defaccr\ecmc{\texttt{ECMC}}
\defaccr\qcec{\textmd{QCEC}}

\defmath\PG{\mathcal{\hat P}_n}
\defmath\hP{{\hat P}}
\defmath\hQ{{\hat Q}}

\defmath\bell{\ket{\mathrm{Bell}}}
\defmath\pauli{\mathcal P}

\usetikzlibrary{intersections,fit,shapes.misc, decorations.markings,patterns}

\makeatletter
\def\namedlabel#1#2{\begingroup
    #2%
    \def\@currentlabel{#2}%
    \phantomsection\label{#1}\endgroup
}
\makeatother

\makeatletter
\patchcmd{\ALG@step}{\addtocounter{ALG@line}{1}}{\refstepcounter{ALG@line}}{}{}
\newcommand{\ALG@lineautorefname}{Line}
\makeatother

\renewcommand\phi{\varphi}

\let\set\undefined

\providecommand{\tuple}[1]{\ensuremath{\left( #1 \right)}}
\providecommand{\set}[1]{\ensuremath{\left\lbrace #1 \right\rbrace}}

\defmath{\bool}{\ensuremath{\mathbb{B}}}
\defmath{\complex}{\ensuremath{\mathbb{C}}}
\defmath{\integers}{\ensuremath{\mathbb{Z}}}
\defmath{\conditionalind}{\mathrel{\text{\scalebox{1.07}{$\perp\mkern-10mu\perp$}}}}
\defmath{\dx}{\partial x}
\defmath{\ddx}{\sfrac{\partial}{\partial x}}
\defmath{\half}{\textstyle{\frac{1}{2}}}

\defmath\Exists{\mathit{Exists}}
\defmath\PlusExists{\mathit{PlusExists}}
\defmath\calciso{\mathsf{calciso}}

\newcommand{\defn}{\,\triangleq\,}

\tikzstyle{oval} = [state, ellipse, minimum size=4mm, inner sep=0.5mm, node distance=1cm]

\tikzstyle{leaf}=[draw, rectangle,minimum size=5mm, inner sep=3pt]
\tikzstyle{var}=[circle,draw=black!70,solid,thick,minimum size=6mm]
\tikzstyle{bdd}=[regular polygon, regular polygon sides=3, draw=black!70,solid,thick,inner sep=0.5mm]
\tikzstyle{n}=[->,loosely dashed,thick]
\tikzstyle{p}=[->,solid,thick]
\tikzstyle{b}=[->,densely dashdotted,ultra thick]

\defmath\before{\prec}
\defmath\beforeq{\preccurlyeq}

\newenvironment{smallmat}{\left[\begin{smallmatrix}}{\end{smallmatrix}\right]}

\defmath\oh{\mathcal O}

\defmath\gmax{g}
\defmath\kmax{\kappa^{\textnormal{final}}}

\defmath\cast{\mathbb C^\ast}

\DeclareRobustCommand{\leafnode}[1][]{%
  \raisebox{-.8mm}{%
  \tikz{%
    \node[state,inner sep=0pt,minimum size=10pt,right= of x,leaf](v){\scriptsize $1$};%
  }%
  }%
}

\newcommand\mat[1]{\begin{bmatrix*}[r]
                    #1
                    \end{bmatrix*}}

\defmath\ww{\begin{smallmat}
      0 & y   \\
      y^* & 0  \\
  \end{smallmat}
}

\def\strictsuccinctto{
    \setbox0\hbox{
            $\longrightarrow$
    }\copy0\llap{\raise\ht0\hbox{
    {
    $    \hspace{0mm}\mathclap{\longleftarrow}{\hspace{-1.5mm}\times}\hspace{0mm}$
    }
    }}
}

\newcommand{\noinfo}{\ding{53}}

\begin{document}

\title{Equivalence~Checking~of\\Quantum~Circuits~by~Model~Counting}

\authorrunning{J. Mei, T. Coopmans, M. Bonsangue \& A. Laarman}
\author{Jingyi Mei, Tim Coopmans, Marcello Bonsangue and Alfons Laarman}
\institute{\vspace{-2em}}

\maketitle

\begin{abstract}
Verifying equivalence between two quantum circuits is a hard problem, that is nonetheless crucial in compiling and optimizing quantum algorithms for real-world devices.
This paper gives a Turing reduction of the (universal) quantum circuits equivalence problem to weighted model counting (WMC).
Our starting point is a folklore theorem showing that equivalence checking of quantum circuits can be done in the so-called Pauli-basis. We combine this insight with a WMC encoding of quantum circuit simulation, which we extend with support for the Toffoli gate.
Finally, we prove that the weights computed by the model counter indeed realize the reduction.
With an open-source implementation, we demonstrate that this novel approach can outperform a state-of-the-art equivalence-checking tool based on ZX calculus and decision diagrams. 
\end{abstract}

\keywords{Quantum computing \and Circuit equivalence \and Satisfiability \and \#SAT \and Weighted model counting \and Pauli basis.\vspace{-1em}}

\section{Introduction}
\label{sec:introduction}

Physicists and chemists regularly deal with `quantum \NP'-hard problems when finding the ground state (energy) of a physical system~\cite{kitaev1997quantum} or assessing the consistency of local density matrices~\cite{liu2006consistency}.
 (the quantum analog of deciding the consistency of marginal probability distributions).
Quantum computing not only holds the potential to provide a matching computational resource for tackling these challenges but also serves as a bridge to incorporate classical reasoning techniques for tackling nature's hardest problems.
Quantum circuits, in particular, offer a precise view into these problems, because the quantum circuit equivalence checking problem is also `quantum \NP'-hard.

Circuit equivalence~\cite{viamontes2007equivalence,wang2008xqdd,yamashita2010fast,amy2018towards,hong2022equivalence,hong2021approximate,berent2022towards,advanced2021burgholzer,thanos2023fast} also has many important applications.
Since quantum computers are highly affected by noise,
it is necessary to optimize the circuits to maximize the performance when running them on a real device.
Furthermore,
many devices can only handle shallow-depth circuits and are subject to various constraints such as connectivity, topology, and native gate sets. 
An essential aspect of designing and optimizing quantum circuits involves verifying whether two quantum circuits implement the same quantum operation.

Equivalence checking for so-called Clifford circuits is  tractable~\cite{thanos2023fast}, which is surprising considering their wide applicability in e.g. quantum error correction~\cite{calderbank1996quantum,steane1996error,shor1995scheme}.
Extending the Clifford gate set with any non-Clifford gate, however,
e.g. with a $T$ or Toffoli gate,
makes the problem immediately `quantum \NP'-hard, that is: \NQP-hard to compute exactly~\cite{tanaka2010exact} and \QMA-hard to approximate~\cite{janzing2005non}, even for constant-depth circuits~\cite{ji2009non}.\footnote{A similar ``jump'' in hardness was noted for quantum circuit simulation in~\cite{van2010classical}. %
}
The exact formulation of equivalence checking allows its discretization~\cite{kissinger_simulating_2022}, exposing the underlying combinatorial problem that classical reasoning methods excel in.
Indeed, exact reasoning methods based on decision diagrams are even used to compute the approximate version of the problem (see e.g. \cite{wei2022accurate,9586214}).

Our aim is to use reasoning tools based on satisfiability (SAT) for \emph{exact} equivalence checking of \emph{universal} quantum circuits.
Like SAT solvers~\cite{biere2009handbook,feng2023verification},
model counters, or \#SAT solvers, can handle complex constraints from industrial-scale applications~\cite{oztok2015atopdown,sang2004combiningcc},
despite the \#P-completeness of the underlying problem.

We propose a new equivalence-checking algorithm based on weighted model counting (WMC).
To do so, we generalize the WMC encoding of quantum circuit simulation from~\cite{qcmc}, showing that it essentially only relies on expressing quantum information in the so-called Pauli basis~\cite{gay2011stabilizer}, thus obviating the need for the arguably more complex stabilizer theory~\cite{gottesman1997stabilizer,efficient2021zhang}.
In addition, we extend the encoding with support for the (non-Clifford) Toffoli gate, allowing more efficient encodings for many circuits.
We then prove that a folklore theorem on quantum circuit equivalence checking~\cite{thanos2023fast} enables the reduction of the problem to a sequence of  weighted Boolean formulas that can be solved using existing weighted model counters (provided they support negative weights~\cite{qcmc}).

We show how the WMC encoding satisfies the conditions of the theorem from \cite{thanos2023fast} and implement the proposed equivalence checking algorithm in the open-source tool \ecmc, which uses the weighted model-counting tool GPMC~\cite{hashimoto2020gpmc}.\footnote{While the theorem presented in \cite{thanos2023fast} already supported universal circuits, the provided tool implementation in \cite{thanos2023fast} is limited to (non-universal) Clifford circuits.}
To assess the scalability and practicality of  \ecmc, we conduct experimental evaluations using random Clifford+$T$ circuits which closely resemble quantum chemistry applications~\cite{wright2022chemistry} and various quantum algorithms from the MQT benchmark~\cite{mqt2023quetschlich}, which includes important quantum algorithms such as QAOA, W-state, and VQE among others.
We compare the results of our method against that of the state-of-the-art circuit equivalence checker \qcec~\cite{advanced2021burgholzer}, showing that in several cases the WMC approach used by our \ecmc tool is competitive.

In summary, this paper provides a many-to-many reduction of (universal) quantum circuit equivalence to weighted model counting (WMC). As a consequence, 
we contribute additional new benchmarks for the WMC competition: basically, each pair of universal quantum circuits can be reduced to a sequence of weighted CNF encodings that need to be solved to (dis)prove equivalence.
This opens up numerous possibilities and challenges to better adapt model counters for this new application area in quantum computing.

\section{General Background}
\label{sec:background}
\defmath\true{1}
\defmath\false{0}
\newcommand\no[1]{\ensuremath{\overline{#1}}}

We provide the necessary background, on quantum computing (\autoref{sec:prelims-quantum}, for a complete introduction see~\cite{nielsen2000quantum}), matrix decomposition in the Pauli basis (\autoref{sec:pauli}) and weighted model counting (\autoref{sec:wmc}).
Readers familiar with any of these topics can skip the corresponding subsection.

\subsection{Quantum Computing}
\label{sec:prelims-quantum}

Quantum computing operates on quantum bits, also called qubits. While a classical bit can be either $0$ or $1$, a qubit has states $\ket{0}=\mat{1,0}^T$ and $\ket{1}=\mat{0,1}^T$. 
Any single-qubit state $\ket{\psi}$ is written as $\alpha\ket{0} + \beta\ket{1}$ where $\alpha$ and $\beta$ are complex numbers with norm $|\alpha|^2 + |\beta|^2 = 1$.
Here $\ket{\psi}$ denotes a pure $n$-qubit quantum state, i.e., an element of the complex vector space
$\complex^{2^n}$.
The states $\ket{0}$ and $\ket{1}$ form the standard computational basis for single-qubit quantum states when $n=1$. For an $n$-qubit quantum system, computational basis states can be obtained by the tensor product of $n$ single-qubit states, e.g., $\ket0 \otimes \ket1 = \ket{01} = \mat{0,1,0,0}^T$.
Only non-entangled states can be written as a tensor product of single-qubit states.
An entangled state is e.g. the well-known Bell state $\nicefrac{1}{\sqrt2}\mat{1, 0,0,1}^T =\nicefrac{1}{\sqrt2} (\ket{00} + \ket{11} )$.
Finally, $\bra{\psi}$ denotes the complex conjugate and transpose of the complex column vector  $\ket{\psi}$, which is a row vector $\bra{\psi} = \ket{\psi}^\dag$. In general, quantum states also have norm~1, which we can now write as $\braket{\psi|\psi} = \bra{\psi}\cdot \ket\psi = 1$.

Quantum states are computed by quantum circuits consisting of quantum gates, i.e. linear operators on the vector space $\mathbb{C}^{2^n}$.
Quantum gates on $n$-qubits can be expressed by $2^n\times 2^n$ \concept{unitary matrices},
i.e., matrices $U \in \complex^{2^n \times 2^n}$ that are invertible and satisfy $U^\dag = U^{-1}$ and hence are norm-preserving as desired.
The so-called Clifford gate set consists of the following two single-qubit and one two-qubit gates.

\vspace{-1em}
\[
  \begin{aligned}
  \text{Hadamard: }
    H = \frac{1}{\sqrt{2}}
    \begin{bmatrix*}[r]
      1 & 1 \\
      1 & -1
    \end{bmatrix*},
   ~
   \text{Phase: }
  S =
  \begin{bmatrix}
    1 & 0 \\
    0 & \dot{\imath}\
  \end{bmatrix},
   ~
  \text{Controlled-$Z$: }
  CZ = 
  \begin{smallmat}
    1 & 0 & 0 & 0 \\
    0 & 1 & 0 & 0 \\
    0 & 0 & 1 & 0 \\
    0 & 0 & 0 & \hspace{-1mm}-1 \\    
  \end{smallmat}
  \end{aligned}
\]
\vspace{-1em}

Though non-universal and classically simulatable~\cite{aaronson2008improved}, Clifford circuits, i.e, circuits composed of Clifford gates only, are expressive enough to describe entanglement, teleportation and superdense coding, and are used in quantum error-correcting codes~\cite{calderbank1996quantum,steane1996error,shor1995scheme} and in measurement-based quantum computation~\cite{raussendorf2001oneway}. Nonetheless, even equivalence checking of Clifford circuits is in $\P$~\cite{thanos2023fast}.

By extending the Clifford gate set with any non-Clifford gate, we immediately obtain a universal gate set, in the sense that arbitrary unitaries can be approximated~\cite{dawson2005solovaykitaev,kitaev1997quantum,kitaev2002classical}.
Below, we list the unitaries for the relevant non-Clifford gates, where $P$ is the phase shift gate and $R_X$ is a rotation gate ($R_Z$ is equivalent to $P$, and we omit $R_Y$).

\vspace{-1em}
\[
  \begin{aligned}
    T = 
    \begin{smallmat}
      1 & 0 \\
      0 & e^{i\frac\pi4}
    \end{smallmat},
  P(\theta) =
  \begin{smallmat}
    1 & 0 \\
      0 & e^{i\theta}
  \end{smallmat},
  R_X(\theta) =
  \begin{smallmat}
    \cos(\nicefrac\theta2) & \hspace{-1mm}-i \sin(\nicefrac\theta2) \\
     \hspace{-1mm}-i \sin(\nicefrac\theta2) & \cos(\nicefrac\theta2)
  \end{smallmat},
  \mathrm{Toffoli} = \hspace{-1mm}
  \begin{smallmat}
    1 & 0 & 0 & 0 & 0 & 0 & 0 & 0 \\
    0 & 1 & 0 & 0 & 0 & 0 & 0 & 0 \\
    0 & 0 & 1 & 0 & 0 & 0 & 0 & 0 \\
    0 & 0 & 0 & 1 & 0 & 0 & 0 & 0 \\    
    0 & 0 & 0 & 0 & 1 & 0 & 0 & 0 \\
    0 & 0 & 0 & 0 & 0 & 1 & 0 & 0 \\
    0 & 0 & 0 & 0 & 0 & 0 & 0 & 1 \\
    0 & 0 & 0 & 0 & 0 & 0 & 1 & 0 \\    
  \end{smallmat}
  \end{aligned}
\]
\vspace{-1em}

We fix $n$ as the number of qubits in the circuits(s) under consideration and write $[m]$ for the set $\set{1, \dots, m}$. Qubits are numbered as $[n]$.

Gates can be composed in parallel (by taking the tensor product of unitaries) and sequentially (by multiplying unitaries, \emph{from right to left}) to obtain circuits.
For instance, $H \otimes S$ represents a two-qubit circuit (unitary) where $H$ acts on the first qubit in parallel with $S$ on the second. We can compose this sequentially with $H \otimes I$, where $I$ is the identity
(no-op), 
by multiplying to the right $(H\otimes S) \cdot (H\otimes I)$. 
Since $H\cdot H = I$, this circuit is equivalent to $I\otimes S$.
For an $n$-qubit quantum system, applying a single-qubit gate $U$ on the $j$-th qubit is represented by $U_j = I^{\otimes j -1}\otimes U \otimes I^{\otimes n- j}$,
for instance: $H_j, X_j, S_j$ (the number of qubits $n$ is implicit here). Similarly, we denote by $CZ_{jk}$ the unitary operator taking the $j$-th qubit as the control qubit and $j$-qubit as the target to execute a controlled-$Z$ gate.
A circuit in our text is simply a list of $n$-qubit unitaries, i.e.,
$C = (G_0, G_1. \dots, G_{m-1})$ where $C$ can in turn be understood as unitary itself $U_C = U_{G_{m-1}} \cdot G_{m-2} \cdots U_{G_0}$.

\begin{example}
Consider circuit $C = (H^{\otimes 2},  CZ_{1,2} ,H_2 )$ below, which computes the Bell state from the all-zero initial state.
The $CZ$ gate is visualized as
{\protect\tikz[baseline=.1ex]{
\protect\fill (0,0) circle (1.5pt) coordinate (A);
\protect\fill (0,1.5ex) circle (1.5pt) coordinate (B);
\protect\draw [line width=.2mm] (A)--(B);}
\hspace{-.2cm}
}
. 
Applying the first parallel Hadamard gates, we obtain $\ket{\phi_1} = H^{\otimes 2}\cdot \ket{00} =\ket{++}$, where the $\ket +$ is the uniform \concept{superposition}, i.e., $\ket+=\nicefrac1{\sqrt2} (\ket 0 + \ket 1)$,
and so is $\ket{++} = \ket+\otimes \ket+ = \frac12(\ket{00}+\ket{01}+\ket{10} + \ket{11}) = \mat{ \frac12& \frac12& \frac12& \frac12 }^T$.
Applying  $CZ$ next yields
$\ket{\phi_2} = \frac12(\ket{00}+\ket{01} + \ket{10} - \ket{11})$.
Finally, by $H_2 = I\otimes H$, we obtain $\nicefrac1{\sqrt2}(\ket{00} + \ket{11})$.
\vspace{-1em}
\begin{center}
\scalebox{.8}{
    \begin{quantikz}[row sep={0.1cm}]
    \lstick{$\ket 0$} \slice{$\ket{\phi_0}$}  & \gate{H}\slice{$\ket{\phi_1}$} 
                & \ctrl{1}\slice{$\ket{\phi_2}$} &\qw \slice{$\ket{\phi_3}$} & \qw \\
     \lstick{$\ket 0$} &\gate{H} & \control{} & \gate{H} & \qw
    \end{quantikz}
}
\end{center}
\vspace{-2em}
\qed
\end{example}

\subsection{Decomposing Matrices in the Pauli Basis}
\label{sec:pauli}

The set of $2^n \times 2^n$ complex matrices %
satisfies the axioms of a vector space.
In particular, this means matrices can be decomposed in a \emph{basis}.

\begin{example}
\label{ex:matrix-basis}
The standard basis for the set of $2\times 2$ complex matrices is
\[
\dyad{0} = \begin{smallmat} 1&0\\0&0\end{smallmat}
,
\quad
|0\rangle\langle 1|
=
\begin{smallmat} 0&1\\0&0\end{smallmat}
,
\quad
|1\rangle\langle 0|
=
\begin{smallmat} 0&0\\1&0\end{smallmat}, 
\quad
\dyad{1}
=
\begin{smallmat} 0&0\\0&1\end{smallmat}
.
\]
The unique decomposition of the matrix $\begin{smallmat} 1&4+i\\4-i&-5\end{smallmat}$ in this basis is
\[
\begin{smallmat} 1&4+i\\4-i&-5\end{smallmat}
=
1 \cdot \begin{smallmat} 1&0\\0&0\end{smallmat} + (4+i) \cdot \begin{smallmat} 0&1\\0&0\end{smallmat}
+ (4-i) \cdot \begin{smallmat} 0&0\\1&0\end{smallmat} + (-5) \cdot \begin{smallmat} 0&0\\0&1\end{smallmat}.
\]

\vspace{-2.5em}
\qed
\end{example}
\vspace{-0.5\baselineskip}

For the choice of basis in \autoref{ex:matrix-basis}, one immediately observes that the decomposition coefficients are indeed $1, 4+i, 4-i$ and $-5$.
For a different basis choice, this is not so obvious.
In general, to determine the coefficients we use the Frobenius inner product for matrices, defined as $\langle M, N \rangle \defn \Tr(M^{\dagger} \cdot N)$ for $2^n \times 2^n$ matrices $M$ and $N$, where the trace $\Tr(M)$ is the sum of the diagonal entries of $M$.
As usual in linear algebra, we say that a basis $\mathfrak{B} = \{B_1, \dots, B_{2^n}\}$ for the $2^n \times 2^n$ matrix space is orthonormal if for all $j, k \in [2^n]$, $\langle B_j, B_k\rangle = 1$ if $j=k$ and $0$ otherwise.
The coefficients are now found as follows:
\spnewtheorem{fact}{Fact}{\bfseries}{\itshape}
\begin{fact}
    Given $2^n\times 2^n$ matrix $M$ and orthonormal basis $\mathfrak{B} = \{B_1, \dots, B_{2^n}\}$ of matrices, the coefficients $\gamma_j \in \mathcal{C}$ of the unique decomposition $M = \sum_{j=1}^{2^n} \gamma_j \cdot B_j$ are found as $\gamma_j = \langle B_j, M\rangle = \Tr(B_j^{\dagger} \cdot M)$.
\end{fact}

\vspace{-0.5\baselineskip}
\begin{example}
\label{ex:matrix-basis-2}
The standard basis from \autoref{ex:matrix-basis} is orthonormal, which one checks by computing the inner products of each pair of basis elements:
\[
\Tr(
\begin{smallmat} 1&0\\0&0\end{smallmat}^{\dagger}
\cdot
\begin{smallmat} 1&0\\0&0\end{smallmat}
)
=
\Tr(
\begin{smallmat} 1&0\\0&0\end{smallmat}
)
=1, \text{ while}
\]
\[
\Tr(
\begin{smallmat} 1&0\\0&0\end{smallmat}^{\dagger}
\cdot
\begin{smallmat} 0&1\\0&0\end{smallmat}
)
=
\Tr(
\begin{smallmat} 1&0\\0&0\end{smallmat}^{\dagger}
\cdot
\begin{smallmat} 0&0\\1&0\end{smallmat}
)
=
\Tr(
\begin{smallmat} 1&0\\0&0\end{smallmat}^{\dagger}
\cdot
\begin{smallmat} 0&0\\0&1\end{smallmat}
)
= 0
\quad\text{etc.}
\]
The coefficients of the matrix $M = \begin{smallmat} 1&4+i\\4-i&-5\end{smallmat}$ from \autoref{ex:matrix-basis} can be found as
\begin{align*}
\Tr(
\begin{smallmat} 1&0\\0&0\end{smallmat}^{\dagger}
\cdot
\begin{smallmat} 1&4+i\\4-i&-5\end{smallmat}
)
&=
\Tr(
\begin{smallmat} 1&4+i\\0&0\end{smallmat}
)
=1,\quad\quad
&\Tr(
\begin{smallmat} 0&1\\0&0\end{smallmat}^{\dagger}
\cdot
\begin{smallmat} 1&4+i\\4-i&-5\end{smallmat}
)
&=4+i,
\\
\Tr(
\begin{smallmat} 0&0\\1&0\end{smallmat}^{\dagger}
\cdot
\begin{smallmat} 1&4+i\\4-i&-5\end{smallmat}
)
&=4-i
, \quad\quad
&\Tr(
\begin{smallmat} 0&0\\0&1\end{smallmat}^{\dagger}
\cdot
\begin{smallmat} 1&4+i\\4-i&-5\end{smallmat}
)
&=-5
.
\end{align*}

\vspace{-2.5em}
\qed
\end{example}
\vspace{-0.5\baselineskip}

In this work, we will express matrices not in the standard basis but in the basis of Pauli strings.
The $2 \times 2$ \emph{Pauli matrices} $X,Y,Z$, together with identity, are
\[
  \begin{aligned}
    \sigma[00] \equiv \,I\, \equiv
    \begin{bmatrix}
      1 & 0 \\
      0 & 1
    \end{bmatrix},~
    \sigma[01] \equiv Z \equiv
    \begin{bmatrix*}[r]
      1 & 0 \\
      0 & -1
    \end{bmatrix*},~
    \sigma[10] \equiv X \equiv
    \begin{bmatrix*}[r]
      0 & 1 \\
      1 & 0
    \end{bmatrix*},~
    \sigma[11] \equiv Y \equiv
    \begin{bmatrix*}[r]
      0 & -\dot{\imath} \\
      \dot{\imath} & 0
    \end{bmatrix*}
  \end{aligned}
\] 

Each Pauli matrix $P$ is unitary ($P^{\dagger} = P^{-1}$)
and involutory ($P\cdot P = I$), and hence also \concept{Hermitian} ($P=P^{\dagger}$).
For $n$ qubits, we define the set of ``Pauli strings'' $\PG \defn \{P_1\otimes \ldots \otimes P_n \mid P_j \in \{ I, X, Y, Z \} \}$. 
Inheriting the properties of Pauli matrices, Pauli strings 
are unitary, involuntary and Hermitian.
It is well-known that the scaled Pauli strings $\{ \frac{1}{\sqrt{2}^n} \cdot P \mid P \in \PG\}$ are an orthonormal basis for the set of $2^n \times 2^n$ complex matrices~\cite{jones2024decomposing}.
Therefore, \autoref{fact:matrix-decomposition} holds. 
\emph{In general, the coefficients $\gamma_P$ are complex numbers, but for Hermitian matrices, they are real.}
\begin{fact}
\label{fact:matrix-decomposition}
We can decompose any $2^n \times 2^n$ complex matrix $M$ as $M = \sum_{P \in \PG} \gamma_P \cdot P$ where the \concept{Pauli coefficient} $\gamma_P = \frac{1}{2^n} \Tr(P^{\dagger} \cdot M)$.
\end{fact}

\begin{lemma}
\label{lemma:hermitian-has-real-pauli-coefficients}.
If $M$ is a $2^n \times 2^n$ Hermitian matrix for some integer $n>0$, then the Pauli coefficient $\frac{1}{2^n} \Tr(P^{\dagger} \cdot M) \in \mathbb{R}$ for each length-$n$ Pauli string $P$.
\end{lemma}
\begin{proof}
We prove that $\Tr(M \cdot P)$ is invariant under complex conjugation.
\begin{align*}
\Tr(M \cdot P)^*\hspace{-1mm} 
&=
\Tr( \left(M \cdot P\right)^* ) 
\\
&= \Tr((\left( M \cdot P\right)^{*})^{\top}) 
&\textnormal{(trace is invariant under transposition)}
\\
&= \Tr(\left( M \cdot P\right)^{\dagger}) 
\\
&=
\Tr(P^{\dagger} \cdot M^{\dagger}) 
\\
&=
\Tr(P \cdot M) 
&\textnormal{($P, M$ are Hermitian)}
\\
&=
\Tr(M \cdot P).
&\textnormal{(by cyclicity of the trace)}
\end{align*}
\vspace{-3.4em}

\phantom X
\end{proof}

\vspace{-0.5\baselineskip}
\begin{example}
The matrix $M = \begin{smallmat} 1&4+i\\4-i&-5\end{smallmat}$ from \autoref{ex:matrix-basis} is Hermitian.
We calculate the coefficients, which are real numbers indeed:
\begin{eqnarray*}
&
\frac{1}{2^1}
\Tr(
I^{\dagger} \cdot M
)
=
\frac{1}{2}
\Tr(
M
)
=
\frac{1 - 5}{2}
=
-2
,
\quad
&
\frac{1}{2^1}
\Tr(
Z^{\dagger} \cdot M
)
=
\frac{1}{2}
\Tr(
\begin{smallmat} 1&4+i\\-4+i&5\end{smallmat}
)
=
3
,
\\
&
\frac{1}{2^1}
\Tr(
X^{\dagger} \cdot M
)
=
4
,
&
\frac{1}{2^1} 
\Tr(
Y^{\dagger} \cdot M
)
=
-1
\end{eqnarray*}
It is straightforward to verify that these are $M$'s Pauli coefficients:
\[
-2  I
+
4  X
-1  Y
+
3 Z
=
-2 \cdot \begin{smallmat}1&0\\0&1\end{smallmat}
+
4 \cdot \begin{smallmat}0&1\\1&0\end{smallmat}
-1 \cdot \begin{smallmat}0&-i\\i&0\end{smallmat}
+
3 \cdot \begin{smallmat}1&0\\0&-1\end{smallmat}
= \begin{smallmat} 1&4+i\\4-i&-5\end{smallmat}
.
\]

\vspace{-2.2em}
\qed
\end{example}

\subsection{Weighted Model Counting}
\label{sec:wmc}

In this work, we will encode the Pauli coefficients of specific matrices as weighted model counting: a sum of weights over all satisfying assignments of a boolean formula.
We here formally describe weighted model counting.

For boolean variables $x,y \in\bool = \set{0,1}$, we define a literal as e.g. $x$ and $\overline x$
and write conjunctions of literals (cubes) as products, e.g., $x\overline y = x \land \overline y$.
A clause is a disjunction of literals, e.g., $\overline x \lor y$.
A formula in conjunctive normal form (CNF) is a conjunction of clauses.
With a weight function, we obtain weighted CNF.

We assign weights to literals using weight function $W\colon \set{ \no x, x \mid x\in \vec x} \to \mathbb{R}$.
Given an assignment $\alpha\in\bool^{\vec x}$,
let $W(\alpha) = \prod_{x\in \vec x} W(x = \alpha(x))$.
Let  $F \colon \bool^{\vec x} \to \bool$ be a propositional formula over boolean variables $\vec x \in \bool^n$ and $W$ a weight function.
We define \concept{weighted model counting}~\cite{biere2009handbook,gomes2021model,chavira2008probabilistic,modelcounting2023} as follows.
\[
MC_W(F) \defn   \sum_{\alpha \in  \bool^V} F(\alpha)\cdot W(\alpha)
\]
\begin{example}
An an example, consider a formula \( F = b \land c \) over \( \vec x = \tuple{a,b,c} \). There exist two satisfying assignments: \( \alpha_1 = abc \) and \( \alpha_2 = \overline a bc \). 
Suppose a weight function \( W \) is defined as follows: \( W(a) = -2 \), \( W(\overline a) = 3 \), \( W(b) = 1/2 \), \( W(\overline b) = 2 \), while \( c \) remains unbiased, i.e., $W(c) = W(\overline c) = 1$. The weighted model counting for \( F \) with respect to \( W \) is computed as follows.
$ MC_W(F) = F(abc) \cdot W( a bc) + F(\overline a bc) \cdot W(\overline a bc)  =  (-2 \cdot \frac12 \cdot 1) + (3 \cdot \frac12 \cdot 1)  = \frac12. $
\qed
\end{example}

\section{Equivalence Checking Circuits in the Pauli Basis}
\label{sec:eqcheck}

In this section, we introduce (exact) equivalence checking~\cite{viamontes2007equivalence,wang2008xqdd,yamashita2010fast,amy2018towards,hong2022equivalence,berent2022towards,advanced2021burgholzer,thanos2023fast} in \autoref{def:eq}, the task we set out to solve.
In this work, we will only consider circuits which consist of gates, and do not contain measurements (this is without loss of generality since measurements be deferred to the end of the circuit~\cite{nielsen2000quantum}).

\begin{definition}\label{def:eq}
    Given two $n$-qubit circuits $U$ and $V$ where $n\in \mathbb{N}^+$,
    $U$ is equivalent to $V$, written $U\equiv V$, if there exists a complex number $c$ such that for all input states $\ket{\psi}$, we have $U\ket{\psi} = cV\ket{\psi}$.
\end{definition}

Here, the factor $c$ is called `global phase' and is irrelevant to any measurement outcome. Given that the circuits are unitaries, it has norm $|c|^2 = 1$.

At first sight, one might think that \autoref{def:eq} requires iterating over all quantum states.
However, although the $n$-qubit quantum state space is continuous, it is a complex vector space of dimension $2^n$, so it suffices to only consider $2^n$ basis vectors for proving $U$ and $V$ equivalent.
In fact, the novice approach to equivalence checking is to decompose $U$ and $V$ in the standard basis (see \autoref{ex:matrix-basis}); that is, to find $U$ and $V$ each by writing each of their individual gates in the standard basis and determining the full unitaries $U$ and $V$ by matrix multiplication (for sequential gates) and tensor product (for parallel gates), and finally checking whether the matrix entries of $U$ are equal to those of $V$, modulo a uniform constant $c$.

Having read \autoref{sec:pauli}, one might expect that we check $U\equiv V$ by checking equality modulo a global constant in the Pauli basis.
While this is certainly possible, this approach has no a priori advantage over the use of the standard basis.
Instead, we will use the following folklore result (for proof see e.g. \cite{thanos2023fast}).

\begin{theorem}\label{thm:main-theorem}
Let $U, V$ be two circuits on $n\geq 1$ qubits.
Then $U$ is equivalent to $V$ if and only if the following condition holds:\\
For all $j\in \{1, 2, \dots, n\}$ and $P\in \set{X, Z}$, we have $U P_j U^{\dagger} = V P_j V^{\dagger}$.
\end{theorem}

The main advantage of using \autoref{thm:main-theorem} instead of directly computing the (matrix entries of the) unitaries $U$ and $V$ is that updating the Pauli coefficients of $UP_j U^{\dagger}$ upon appending a gate to $U$, i.e. the update the coefficients of $UP_j U^{\dagger}$ to those of $(GU) P_j (GU)^{\dagger} = G \left(U P_j U^{\dagger}\right) G^{\dagger}$ for some gate $G$, is computationally easy for the Clifford gates $H, S, CNOT$, having lead to efficient Clifford circuit equivalence checking~\cite{thanos2023fast}.
In this work, we will include $T$ gates, Toffoli gates, and Pauli rotation gates, too, which enable equivalence checking of universal quantum computing (and which lifts the hardness of equivalence checking to quantum analogs of \NP, see \autoref{sec:introduction}).

Another advantage of \autoref{thm:main-theorem} is that, since $U$ is a unitary, $U P_j U^{\dagger}$ is Hermitian, so that its Pauli coefficients are real numbers by \autoref{lemma:hermitian-has-real-pauli-coefficients}, relieving us from the need to use complex numbers.

\begin{example}
\label{ex:eqpauli}
Choose $V = S_1$ and $U = T_1 T_1$.
In order to determine whether $U\equiv V$, we compute the Pauli coefficients of $UXU^{\dagger}, UZU^{\dagger}, VXV^{\dagger}$ and $VZV^{\dagger}$ as follows using \autoref{tab:clifford}.
\begin{table}[h]
    \centering
    \renewcommand{\arraystretch}{1.7}
    \vspace{-1.5em}
\scalebox{0.8}{
    \begin{tabular}{c|c|c|c|c}
       \multicolumn{2}{c|}{} & $\mathbf P^0$ & $\mathbf P^1 = T\mathbf P^0T^\dagger$ & $\mathbf P^2 = T\mathbf P^1T^\dagger$ \\
       \hline
    \multirowcell{2}{
        \scalebox{0.7}{
        \begin{minipage}{.2\textwidth}
    \begin{quantikz}[row sep={0.1cm}]
    \slice{$\mathbf P^0$} &
    \gate{T}\slice{$\mathbf P^1$} \gategroup[1,steps=2,style={dashed,rounded corners,fill=blue!20, inner ysep=0pt, inner xsep=0pt},
    			         background,label style={label position=below,anchor=north,yshift=-0.2cm}]{$C_U$}
        & \gate{T}\slice{$\mathbf P^2$} & 
    \end{quantikz}
    \end{minipage}~~~~
        }}   
        & $UXU^{\dagger}$ & $X$     & $\tfrac{1}{\sqrt{2}}(X+Y)$ & $\tfrac{1}{2}(X+Y+Y-X) = Y$ \\
        \cline{2-5}
        & $UZU^{\dagger}$ & $Z$     & $Z$ & $Z$ \\
        \hline
    \end{tabular}
}
    \qquad
    \scalebox{0.8}{\begin{tabular}{c| c|c|c}
       \multicolumn{2}{c|}{} & $\mathbf P^0$ & $\mathbf P^1 = S\mathbf P^0S^\dagger$ \\
       \hline
    \multirowcell{2}{
    \scalebox{0.7}{
    \begin{minipage}{.2\textwidth}
    \begin{quantikz}[row sep={0.1cm}]
    \slice{$\mathbf P^0$} &
    \gate{S}\slice{$\mathbf P^1$} \gategroup[1,steps=1,style={dashed,rounded corners,fill=blue!20, inner ysep=0pt, inner xsep=0pt},
    			         background,label style={label position=below,anchor=north,yshift=-0.2cm}]{{$C_V$}} & 
    \end{quantikz}
    \end{minipage}\hspace{-.7cm}
    }}   
        & $VXV^{\dagger}$ & $X$     & $Y$ \\
        \cline{2-4}
        & $VZV^{\dagger}$ & $Z$     & $Z$  \\
        \hline
    \end{tabular}
}
    \vspace{-2em}
\end{table}
By \autoref{thm:main-theorem}, this implies that $U$ and $V$ are equivalent, which we verify by writing their unitaries in the standard basis as follows.
\[
U = S = \begin{smallmat} 1&0\\0 & i\end{smallmat}, \quad V = T\cdot T = \begin{smallmat} 1&0\\0&\sqrt{i}\end{smallmat} \cdot \begin{smallmat} 1&0\\0&\sqrt{i}\end{smallmat} = \begin{smallmat} 1&0\\0&i\end{smallmat}
\]
Finally, we remark that $UXU^{\dagger} = \tfrac{1}{2}(\cancel{X}+Y+Y-\cancel{X}) = Y$ represents both constructive ($Y$ terms add up) as well as destructive interference ($X$ terms cancel).
\qed
\end{example}

We will finish this section by explaining the intuition behind \autoref{thm:main-theorem}, by rephrasing its proof from \cite{thanos2023fast}.
The first step in the proof is to realize that \autoref{def:eq} is equivalent to the following.

\begin{lemma}\label{def:eq2}
    Given two $n$-qubit circuits $U$ and $V$ where $n\in \mathbb{N}^+$,
    $U$ is equivalent to $V$ iff for all $n$-qubit quantum states $\ket{\phi}$, we have $U\dyad{\phi} U^{\dagger} = V \dyad{\phi} V^{\dagger}$.
\end{lemma}

The object $\dyad{\phi}$ in \autoref{def:eq2} is the \concept{density matrix} representation of~state~$\ket{\phi}$~\cite{nielsen2000quantum}.
For any unitary $U$, with $\ket\psi = U\ket\phi$, we can easily see that the corresponding operation on the density matrix $\op{\psi}{\psi}$ is an operation called \concept{conjugation}, i.e.,
$\op{\psi}{\psi} = U \op{\phi}{\phi} U^\dagger$.
Density matrices are $2^n \times 2^n$ Hermitian matrices and can thus be expressed as a real-weighted sum of Pauli strings by \autoref{lemma:hermitian-has-real-pauli-coefficients}.
For this reason, we observe that if $UPU^{\dagger} = VPV^{\dagger}$ for each Pauli string $P$, i.e. $U$ and $V$ coincide on all Pauli strings by conjugation, then $U$ and $V$ must also coincide on all density matrices by conjugation, and thus they are equivalent by \autoref{def:eq2}.

The final step in proving \autoref{thm:main-theorem} is to realize that for a unitary matrix, the conjugation action is completely determined by fixing its conjugation action on only all $X_j$ and $Z_j$ for $j\in [n]$.
This insight relies on two parts: First, each Pauli string can be written as the product of $X_j$ and $Z_j$ modulo a factor $\in \{\pm 1, \pm i\}$.
Second, for a unitary $M$, we have $M^{\dagger} M = \unit$, which implies that instead of first multiplying $X_j$s and $Z_j$s to construct a Pauli string, followed by conjugation, one can first conjugate and subsequently multiply to arrive at the same result.\footnote{The conjugation map $P \mapsto UP U^{\dagger}$ is a group isomorphism.}
For example, $M X_j M^{\dagger} \cdot M Z_j M^{\dagger} = M X_j \unit Z_j M^{\dagger} = M X_j Z_j M^{\dagger}$.

We finish this section by observing that in \autoref{tab:clifford}, the last two non-Clifford gates yield a linear combination of Pauli strings~\cite{qcmc} for each Pauli string (matrix).
This potentially causes an explosion of the number of Pauli strings when conjugating multiple non-Clifford gates.
To handle this, we will exploit the strength of model counters in \autoref{sec:encoding} by representing Pauli strings $\hP$ as satisfying assignments which are weighted by the coefficient $\gamma_\hP$, as explained next in \autoref{sec:wmc}.

\section{Encoding Quantum Circuit Equivalence in SAT}
\label{sec3}
\label{sec:encoding}

\begin{table}[t!]
  \centering
  \caption{Lookup table for conjugating Pauli gates by Clifford+$T$+$R_X$ gates. The subscripts ``c'' and ``t'' stand for ``control" and ``target". Adapted from~\cite{mei2024simulating}.}
  \label{tab:clifford}
  \setlength{\tabcolsep}{3pt} 
  \begin{tabularx}{0.92\textwidth}{c|rr||c|cc||c|cc}
      \toprule
      \textbf{Gate} & \textbf{In} & \textbf{Out} & \textbf{Gate} & \textbf{In} & \textbf{Out} & \textbf{Gate} & \textbf{In} & \textbf{Out} \\
      \midrule
      & $X$ & $Z$ & \multirow{6}{*}{CZ} & $\phantom{-}I_c \otimes X_t$ & $\phantom{-}Z_c\otimes X_t$ && $X$& $\frac1{\sqrt2}(X+Y)$ \\
      $H$ & $Y$ & $-Y$ & & $\phantom{-}X_c \otimes I_t$ & $\phantom{-}X_c \otimes Z_t$ & $T$ & $Y$ & $\frac1{\sqrt2}(Y - X)$ \\
      & $Z$ & $X$ & & $\phantom{-}I_c \otimes Y_t$ & $\phantom{-}Z_c \otimes Y_t$   &&$Z$ & $Z$ \\
      \cline{1-3} \cline{7-9}
       & $X$ & $Y$ & & $\phantom{-}Y_c \otimes I_t$ & $\phantom{-}Y_c \otimes Z_t$  &&$X$& $X$  \\
      $S$ & $Y$ & $-X$ & & $\phantom{-}I_c \otimes Z_t$ & $\phantom{-}I_c \otimes Z_t$  &$R_X(\theta)$& $Y$& $\cos(\theta)Y + \sin(\theta) Z$ \\
      & $Z$ & $Z$ & & $\phantom{-}Z_c \otimes I_t$ & $\phantom{-}Z_c \otimes I_t$ && $Z$& $\cos(\theta)Z - \sin(\theta) Y$\\
      \bottomrule
  \end{tabularx}
  \vspace*{-1em} 
\end{table}

The previous section \autoref{sec:eqcheck}, centered around \autoref{thm:main-theorem}, explained that equivalence checking can be done by conjugating Pauli strings with unitaries, and that the required calculations for this approach are the same as in simulation of quantum circuits using a density matrix representation of the quantum state.
In this section, we show how we reduce equivalence checking of universal quantum circuits to weighted model counting, which is formalized in \autoref{prop:wmc} below.
Our approach is based on the $\mathcal O(n +m)$-length encoding for quantum circuit simulation provided in \cite{qcmc}.
Finally, our encoding in this work extends \cite{qcmc} with Toffoli gates.

For the rest of the paper,
we use $P$ for an unweighted Pauli string and
we use $\mathbf P$ for a summation of weighted Pauli strings,
e.g. $\tfrac{1}{\sqrt{2}}X + \tfrac{1}{\sqrt{2}}Y$.
We will also use the words `circuit' and `unitary' interchangeably, because it is clear from the context which two is meant (recall that we only consider circuits without measurements).

For notational simplicity, we will solve a rephrased version of the equivalence checking problem from \autoref{def:eq} in \autoref{sec:eqcheck}: to check whether a unitary $A$ is equivalent to the identity unitary $\unit$, which leaves every input unchanged.
By choosing $A \defn V^{\dagger}U$, we see that $U \equiv V$ precisely if $A \equiv \unit$.
If $U$ and $V$ consist of gates $U = (U_0, U_1, \dots, U_{m-1})$ and $V = (V_0, V_1, \dots, V_{\ell-1})$ for $m, \ell \in \mathbb{N}_{\geq 1}$, then a circuit for $A$ is given as the $m + \ell$ gates $A = (U_0, U_1, \dots, U_{m-1}, V_{\ell-1}^{\dagger}, V_{\ell - 2}^{\dagger}, \dots, V_{0}^{\dagger})$.

Following \autoref{thm:main-theorem}, our approach will be as follows: given a unitary $A$ and its circuit $C_A =\tuple{G_0\cdots G_{m-1}} \in \set{H_j,S_j, CZ_{jk}, T_j, \textrm{Toffoli}_{jkl}, R_X(\theta)_j, \dots \mid j,k,l \in [n]}^m$, we need to compute $\mathbf P^m \defn A \mathbf{P}^0 A^{\dagger}$ from an initial $\mathbf P^0\in\{+X_i, +Z_i \mid i\in[n]\}$, showing that $\mathbf P^m = \mathbf P^0$.
Since $\textbf P^0$ is a Pauli string and thus Hermitian, so~is~$\textbf P^m$.
Based on \autoref{lemma:hermitian-has-real-pauli-coefficients}, we can track each $\mathbf P^k$ as a real-weighted sum of Pauli strings.

\subsection{Encoding Pauli coefficients as Weighted Model Counts}
We first explain the encoding for circuit simulation from \cite{qcmc}, where we encode the real-weighted sum of Pauli operators $\mathbf P$ and the update rules of the circuit $C$ as weighted boolean formulas.
We start with the simplest case --- a Pauli string, 
then consider how to encode a single summand, i.e., a single weighted Pauli operator, 
in the end extend this to a weighted sum of Pauli operators.

Given a Pauli string $P=\bigotimes_{i\in[n]}\sigma[a_i, b_i]$ with $a_i, b_i \in \{0,1\}$,
the corresponding encoding is denoted as $F_P$, which is the boolean formula which only has $\{x_1\gets a_1,\cdots, x_n\gets a_n, z_1\gets b_1,\cdots, z_n\gets b_n\}$ as satisfying assignment, for example $F_{Z \otimes X} = F_{\sigma[01]\otimes \sigma[10]} = \no x_1 x_2 z_1 \no z_2$.
When it comes to weighted Pauli string,
although the weights are never imaginary,
they can still have a $\pm$ sign.
A weighted Pauli operator can be therefore encoded by $2n+1$ boolean variables:
two bits $x_i,z_i$ for each of the $n$ Pauli matrices and one sign bit $r$,
such that $\mathbf P =(-1)^r\sigma[x_{1},z_{1}] \otimes \ldots \otimes \sigma[x_{n},z_{n}]$.
For example, consider boolean formula $F_{\mathbf P} = r\no x_1 z_1x_2z_2$ where $\mathbf{P} = -Z\otimes Y$. Its one satisfying assignment is $\{r\gets 1, x_1 \gets 0, z_1\gets 1, x_2 \gets 1, z_2\gets 1\} \equiv -Z\otimes Y$.
We later introduce weights $W(r) = -1$ and $W(\no r) = 1$ to interpret the sign.
So for a formula $F(x_1,z_1, \dots,x_n,z_n,r)$, we let the satisfying assignment represent a set (sum) of Pauli strings.
The base case is the formula $F_{\mathbf P^0} = F_{P}$ for a Pauli string $P \in \{X_j, Z_j \mid j \in [n]\}$.
Next, we need to encode how sums of Pauli operators evolve
when conjugating with the gates of the circuit, one by one.
For this, our encoding duplicates the variables for all $m$ gates (each time step) as follows (which is similar to encodings for bounded model checking~\cite{biere2009bounded}).
\[
  \vec w^t = \{x_j^t, z_j^t, r^t \mid j\in[n]\}\text{  for $t\in\set{0, \dots, m}$}
  \text{~~~~and~~~~}  \vec v^t = \bigcup_{i\in[t]} \vec w^i.
\]
For example, $\mathbf P^0(\vec w^0) = X_1$ is encoded as $\no r^0 x_1^0 \no z^0_1 \dots x_n^0 \no z^0_n $.
Also, the satisfying assignments of a constraint $C(\vec v^m)$ projected to variables $\vec w^t$ represent the sum of Pauli operators after conjugating the initial $t$ gates $G_1, \dots, G_t$ of the circuit~$A$, written:
\[ 
C(\vec v^m)[\vec w^t] = \sum_{\alpha\in \set{0,1}^{\vec v^m}} C(\alpha) \cdot(-1)^{\alpha(r^t)}\,\cdot \bigotimes_{j\in[n]}\sigma[\alpha(x^t_{j}),\alpha(z^t_{j})]
\]
The next question is how to encode gate semantics, i.e., define a constraint to get $\mathbf P^1$ by conjugating gate $G_0$ to $\mathbf P^0$, etc.
Note that since $\mathbf P^0 \in \set{X_j, Z_j \mid j \in [n]}$ consists of a sum of only one Pauli operator, for Clifford circuits $C$, there will only be a single satisfying assignment~$\alpha$ for all time steps $t\in[m]$, since e.g. $H X H^\dag = Z$ (and not e.g. $Z + Y$).
Non-Clifford gates, like $T$ or Toffoli, will add satisfying assignments representing summands with different weights (e.g. powers of  powers of $\nicefrac 1{\sqrt2}$ for the $T$ gate as discussed above).
To encode these weights, 
we introduce new variables $u^t$, but only at time steps $t$ with a $T$ gate (i.e., $G_t= T$).

When a gate $T_j$ is performed and there is a satisfying assignment with $x^t_{j}=1$, it means that we are conjugating a $T$ gate on the $j$-th qubit set to $\pm X$ or $\pm Y$ and the result should be
either $TXT^\dag = \frac{1}{\sqrt{2}}(X+Y)$ or $TYT^\dag=\frac{1}{\sqrt{2}}(Y-X)$ (modulo sign).
To achieve this the encoding should let $z^{t+1}$ unconstrained and setting
$u^t\Leftrightarrow x^t_j$.
Accordingly, we set the weights $W(u^t) = \frac1{\sqrt 2}$ and
$W(\no u^t) = 1$.
\autoref{tab:booltgate} illustrates  how the boolean variables $\vec w^t$ and $\vec w^{t+1}$ relate for a $T$ gate (derived by computing $T_j P T_j^\dagger$ for Pauli gate $P$).

The encoding of gate semantics can be derived similarly.
For example the boolean constraint for $H^t_j$ is given by
\[
  F_{H^t_j}(\vec w^{t}, \vec w^{t+1}) \defn  r^{t+1} \Leftrightarrow r^{t} \oplus x^t_{j}z^t_{j}
  \land~  z^{t+1}_{j} \Leftrightarrow  x^t_{j}
  \land~  x^{t+1}_{j} \Leftrightarrow  z^t_{j}
\]
Here we omit additional constraints $a^{t+1} \Leftrightarrow a^t$
for all unconstrained time-step-$t+1$ variables $a$, i.e., for $a = x^{t+1}_l,z_l^{t+1}$ with $l \neq j$.
Similarly, by abbreviating $F_{G_t}(\vec w^t, \vec w^{t+1})$ as $G^t$, the encoding for other Clifford+$T$ gates are as follows:
\begin{align*}%
\
  S^t_j &\defn r^{t+1} \Leftrightarrow r^{t} \oplus x^t_{j}z^t_{j}
  &&~\land~  z^{t+1}_{j} \Leftrightarrow  x^t_{j} \oplus  z^t_{j}   \\
\
  T^t_j &\defn  x^{t+1}_{j} \Leftrightarrow x^t_{j} 
  			\neg x^{t}_{j}  (z^{t+1}_{j} \Leftrightarrow z^t_{j})
  			\hspace{-.4cm}
  &&~\land~ r^{t+1} \Leftrightarrow r^t \oplus  x^t_{j}  z^t_{j}  \neg z^{t+1}_{j} 
  			\hspace{-.5cm}
  &&~\land~ u^t \Leftrightarrow x^t_{j}.\\
  CZ^t_{jk} &\defn r^{t+1} \Leftrightarrow r^{t} \oplus x^t_{j}x^t_{k}
    (z^t_{k} \oplus z^t_{j})
  \ \hspace{-0.6cm}
  &&~\land~  z^{t+1}_{k} \Leftrightarrow  z^t_{j} \oplus  x^t_{k}
  &&~\land~  z^{t+1}_{j} \Leftrightarrow  z^t_{k} \oplus  x^t_{j}
\end{align*}
\vspace{-1.4em}

\begin{table}[t!]
\caption{Boolean variables under the action of conjugating one T gate. Here we omit the sign $(-1)^{r^{t}}$ for all $P$ and sign $(-1)^{r^{t+1}}$ for all ${TPT^\dag}$.}
\label{tab:booltgate}
\setlength{\tabcolsep}{6pt} %
\renewcommand{\arraystretch}{1} %
  \centering
  \begin{tabular}{rc|c|c|c|c|c}
    \hline
     $P$  & $x^t z^t r^t$ & $TPT^\dag$  & $x^{t+1}$ & $z^{t+1}$ & $r^{t+1}$ & $u^t$  \\
     \hline
     $I$    & 00 $r^t$ & $I$ & \multirow[]{2}{*}{ 0} & \multirow[]{2}{*}{$z^t$} & \multirow[]{2}{*}{$r^t$} &\multirow[]{2}{*}{0} \\
\cline{1-3}
    $Z$     & 01 $r^t$ &$Z$  &  & &  \\
    \hline
    $X$    & 10 $r^t$ &$\frac{1}{\sqrt{2}}(X+Y)$ & \multirow[]{2}{*}{1} & \multirow[]{2}{*}{\{0,1\}} & $r^t$ & \multirow[]{2}{*}{$1$} \\
    \cline{1-3}
    \cline{6-6}
    $Y$     & 11 $r^t$ &$\frac{1}{\sqrt{2}}(Y-X)$ & & & $r^t\oplus \neg z^{t+1}$ & \\
    \hline
  \end{tabular}
\end{table}

To this end,
we can inductively define boolean constraints for each time step as
$F_{\mathbf P^{t}}(\vec v^{t}) = F_{\mathbf P^0}(\vec v^0)\wedge \bigwedge_{i\in [t-1]} G^i (\vec v^{i}, \vec v^{i+1})$ for $t\geq 1$,
where $G_i$ denotes the gate at time step $i$ and $F_{\mathbf P^0}(\vec v^0)$ encodes $\mathbf P^0$.

\begin{example}\label{ex:sat-update}
Reconsider the circuit $U = T \cdot T$ from \autoref{ex:eqpauli}.
Starting with $\mathbf P^0 = X$, the formulas are $F_{\mathbf P^0} =  x^0_1 \no{z^0_1} \no{r^0}$, $F_{\mathbf P^1} = F_{\mathbf P^0} \wedge F_{T^0_1}$, i.e.
\begin{align*}
    F_{\mathbf P^1} \hspace{-.5mm}= F_{\mathbf P^0}\hspace{-.5mm} \wedge \hspace{-.5mm} F_{T^0_1} \hspace{-.5mm}= x^0_1 \no{z^0_1} \no{r^0}\wedge   x^{1}_{1} \hspace{-.5mm}\Leftrightarrow x^0_{1} 
  			\neg x^{0}_{1}  (z^{1}_{1} \Leftrightarrow z^0_{1})
  \wedge r^{1} \hspace{-.5mm}\Leftrightarrow r^0 \oplus  x^0_{1}  z^0_{1}  \neg z^{1}_{1} 
  \wedge u^0 \hspace{-.5mm}\Leftrightarrow x^0_{1},
\end{align*}
and similarly $F_{\mathbf P^2} = F_{\mathbf P^1} \wedge F_{T^1_1}$.
\qed
\end{example}

A formalization of the explanation above as induction over the gates proves the following proposition, which relates the weighted model counts to the Pauli coefficients (see \autoref{fact:matrix-decomposition}).

\begin{proposition}
[WMC computes the Pauli coefficients]
\label{prop:encoding}
Let $C_A$ be an $n$-qubit circuit,
$A = G_0, G_1,\cdots G^{m-1}$ the corresponding unitary and
$\mathbf P^0$ a Pauli string,
so that
the encoding of $\mathbf P^m \defn A \mathbf P^0 A^\dag$ is given by $F_{\mathbf{P}^m} \defn F_{\mathbf{P}^0} \wedge \bigwedge_{i\in[m]}F_{G^i}$ with according weight function $W$.
For any $\mathbf{P}^0 \in \{+X_j, +Z_j \mid j \in [n]\}$,
the weighted model count of $F_{\mathbf P^m}\wedge F_{P}$ equals the Pauli coefficient $\gamma_{P}$ of $A\mathbf{P}^0 A^{\dag}$ for all $P \in \PG$.
That is, $MC_W(F_{\mathbf P^m}\wedge F_{P}) = \frac{1}{2^n} \cdot \Tr(P^{\dagger} \cdot A\mathbf P^0 A^{\dagger})$ for all $P \in \PG$.
\end{proposition}

We emphasize the necessity for using negative weights.
For example, in \autoref{ex:eqpauli}, we have $\mathbf{P}^2 = U\mathbf{P}^0 U^{\dagger} = \tfrac{1}{2}(\cancel{X}+Y+Y-\cancel{X}) = Y$ for $\mathbf{P}^0 = X$, where the terms $X$ and $-X$ cancel each other out, while the $Y$ terms add up. 
\emph{This is why weighted model counting with negative weights is required; to reason about such constructive and destructive interference, ubiquitous to quantum computing.}

\begin{example}\label{ex:sat-weight}
Following \autoref{ex:sat-update},
we have the satisfying assignments for $F_{\mathbf P^0}$, $F_{\mathbf P^1}$ and $F_{\mathbf P^2}$ as:
\begin{alignat*}{2}
\vspace{-1em}
  SAT(F_{\mathbf P^0}) &~= \{&&x^0_1\no {z^0_1}\no{r^0_1}\}, \\
  SAT(F_{\mathbf P^1}) &~= \{&&x^0_1\no {z^0_1}\no{r^0_1} \ x^1_1 \no{z^1_1} \no{r^1_1} \ u^0, 
  \ x^0_1\no {z^0_1}\no{r^0_1} \ x^1_1z^1_1\no{r^1_1} \ u^0\}, \\
        SAT(F_{\mathbf P^2}) &~= \{
    && x^0_1\no {z^0_1}\no{r^0_1} \
    x^1_1 \no{z^1_1} \no{r^1_1} \ 
    x^2_1 \no{z^2_1} \no{r^2_1} \ u^0 u^1, \
    x^0_1\no {z^0_1}\no{r^0_1} \
    x^1_1 \no{z^1_1} \no{r^1_1} \
    x^2_1 {z^2_1} \no{r^2_1} \ u^0 u^1, \\
    & && x^0_1\no {z^0_1}\no{r^0_1} \
    x^1_1 {z^1_1} \no{r^1_1} \ 
    x^2_1 {z^2_1} \no{r^2_1} \ 
    u^0 u^1, \
    x^0_1\no {z^0_1}\no{r^0_1} \
    x^1_1 {z^1_1} \no{r^1_1} \
    x^2_1 \no{z^2_1} {r^2_1} \
    u^0 u^1\},
\end{alignat*}
with the weight function $W(r^2_1) = -1$, $W(\no {r^2_1}) = 1$, $W(u^0) = W(u^1) = \frac{1}{\sqrt{2}}$ and $W(\no{u^0}) = W(\no{u^1}) = 1$.
Each of the satisfying assignments corresponds to a term in the Pauli decomposition of $\mathbf{P}^2$, which we recall from \autoref{ex:eqpauli} to be
\begin{equation}
\mathbf P^2 = \tfrac{1}{2}X+\tfrac{1}{2}Y+\tfrac{1}{2}Y-\tfrac{1}{2}X = (\tfrac{1}{2} - \tfrac{1}{2}) X + (\tfrac{1}{2} + \tfrac{1}{2}) Y = Y
.
\label{eq:p2}
\end{equation}
For example, the term $-\tfrac{1}{2} X$ is encoded by $x^0_1\no {z^0_1}\no{r^0_1} \
    x^1_1 {z^1_1} \no{r^1_1} \
    x^2_1 \no{z^2_1} {r^2_1} \
    u^0 u^1$
    because it contains $x^2_1 \no{z^2_1}$ (corresponding to $X$) and its weight is $W(r^2_1) \cdot W(u^0) \cdot W(u^1) = (-1) \cdot \frac{1}{\sqrt{2}} \cdot \frac{1}{\sqrt{2}} = -\frac{1}{2}$.
We verify that the constructive interference of the $Y$ terms in \eqref{eq:p2} (i.e. they add up) results in an aggregate Pauli coefficient $\gamma_Y$ of $\mathbf{P}^2$ of $1$:
\[
    MC_W(F_{\mathbf P^2} \wedge F_Y) =  \tfrac{1}{\sqrt{2}}\cdot\tfrac{1}{\sqrt{2}} + \tfrac{1}{\sqrt{2}}\cdot\tfrac{1}{\sqrt{2}} = 1 = \tfrac{1}{2}\Tr(Y\cdot \mathbf P^2).
\]
Similarly, we verify that destructive interference of the $X$ terms in \eqref{eq:p2} (i.e. they cancel) results in the coefficient $\gamma_X$ being $0$:
\[
    MC_W(F_{\mathbf P^2} \wedge F_X) = \tfrac{1}{\sqrt{2}}\cdot\tfrac{1}{\sqrt{2}} - \tfrac{1}{\sqrt{2}}\cdot\tfrac{1}{\sqrt{2}} = 0 = \tfrac{1}{2}\Tr(X\cdot \mathbf P^2).
\qed
\]
\end{example}

\vspace{-2em}
\subsubsection{Toffoli Gate.}
Similar to the way gate encodings of other non-Clifford gates were derived, we can encode the Toffoli gate. To this end, we brute forced the Toffoli gate behavior in the Pauli domain.
 To keep things readable, we will only present a lookup table in the Pauli basis in \autoref{tab:toffoli}, like \autoref{tab:clifford}. 
The corresponding boolean constraint can easily be derived.
To subsequently obtain a minimal (weighted) CNF formula, we applied the Quine–McCluskey algorithm~\cite{mccluskey1956minimization,quine1952problem}.

\newpage

  \begin{table}[t!]
  \setlength{\tabcolsep}{10pt} %
     \caption{Pauli operator lookup table for the Toffoli gate for in/output $P$ and $Q$, highlighting interesting inputs. \autoref{tab:toffoli-full} in  \hyperref[app:toffoli]{Appendix~\ref*{app:toffoli}} provides the full table.}
     \label{tab:toffoli}
    \centering
    \scriptsize
    \begin{tabular}{ >{\centering\arraybackslash}p{1.8cm} | l }

\toprule
 $P\in \mathcal P_3$ & $Q= \mathrm{Toffoli} \cdot P \cdot  \mathrm{Toffoli}^\dagger $ with $Q\in \frac 12  \sum_{i\in [4]} \mathcal P_3 $ or $Q\in \mathcal P_3$  \\
\midrule
$ I \otimes I \otimes Z $ & $(   I \otimes I \otimes Z +   I \otimes Z \otimes Z +   Z \otimes I \otimes Z - Z \otimes Z \otimes Z ) / 2$ \\
$ I \otimes I \otimes X $ & $   I \otimes I \otimes X $ \\
$ I \otimes Z \otimes I $ & $   I \otimes Z \otimes I $ \\
$ I \otimes X \otimes I $ & $(   I \otimes X \otimes I +   I \otimes X \otimes X +   Z \otimes X \otimes I - Z \otimes X \otimes X ) / 2$ \\
$ Z \otimes I \otimes I $ & $   Z \otimes I \otimes I $ \\
$ X \otimes I \otimes I $ & $(   X \otimes I \otimes I +   X \otimes I \otimes X +   X \otimes Z \otimes I - X \otimes Z \otimes X ) / 2$ \\
$ X \otimes X \otimes Z $ & $(   X \otimes X \otimes Z +   Y \otimes Y \otimes Z - X \otimes Y \otimes Y - Y \otimes X \otimes Y ) / 2$ \\
$ Y \otimes Y \otimes Y $ & $(   X \otimes X \otimes Y +   Y \otimes Y \otimes Y - X \otimes Y \otimes Z - Y \otimes X \otimes Z ) / 2$ \\

\bottomrule
    \end{tabular}
  \end{table}

\subsection{Weighted Model Counting-based Algorithm for Equivalence Checking}
The previous subsection explains how to encode the Pauli coefficients of $A P A^{\dagger}$, where $A$ is a unitary and $P$ a Pauli string, in a boolean formula together with a weight function.
We here connect this encoding to \autoref{thm:main-theorem}, which expresses that determining whether a unitary $A$ is equivalent to the identity circuit can be done by checking if $APA^{\dagger} \stackrel{?}{=} P$ for Pauli strings $P \in \{X_j, Z_j \mid j \in [n]\}$.
We use the following lemma, which expresses that for any unitary $A$ and Pauli string $P$, the $P$-Pauli coefficient of $AP A^{\dagger}$ can only become $1$ if $A P A^{\dagger}$ equals $P$.

\begin{lemma}\label{th:wmc}
Let $A$ be a unitary and $P \in \PG$ be any (Hermitian) Pauli operator.
Then $AP A^{\dagger} = P$ if and only if $\frac{1}{2^n} \Tr(A P A^{\dagger} \cdot P) = 1$.
\end{lemma}
\begin{proof}
If $AP_j A^{\dagger} = P_j$, then $\Tr(A P_j A^{\dagger} \cdot P_j) = \Tr(P_j \cdot P_j) = \Tr(I^{\otimes n}) = 2^n$.
For the converse direction, we observe that $\Tr(A P_j A^{\dagger} \cdot P_j)$ is the Frobenius inner product
$\Tr(U^{\dagger} V)$ for $U \defn A P_j A^{\dagger}$ and $V \defn P_j$.
It now follows from the Cauchy-Schwarz inequality $|\langle U,  V\rangle |^2 \leq \langle U,U \rangle \cdot \langle V ,V \rangle$ that
\begin{align*}
|\Tr(A P_j A^{\dagger} \cdot P_j)|^2
&\leq
\Tr((A P_j A^{\dagger})^{\dagger} \cdot A P_j A^{\dagger})
\cdot
\Tr(P_j^{\dagger} \cdot P_j)
\\
&=
\Tr(A P_j A^{\dagger} \cdot A P_j A^{\dagger})
\cdot
\Tr(P_j^{\dagger} \cdot P_j)
& %
\\
&=
\Tr(A P_j \cdot P_j A^{\dagger})
\cdot
\Tr(I^{\otimes n})
&\textnormal{($A$ and $P_j$ are unitary)}
\\
&=
\Tr(A A^{\dagger})
\cdot
\Tr(I^{\otimes n})
&\textnormal{($P^2 = I$ for all $P \in \PG$)}
\\
&=
\Tr(I^{\otimes n})
\cdot
\Tr(I^{\otimes n})
&\textnormal{($A$ is unitary)}
\\
&=
2^n \cdot 2^n = 4^n
\end{align*}
and therefore
$|\Tr(A P_j A^{\dagger} \cdot P_j)| \leq 2^n$.
Since $\Tr(A P_j A^{\dagger} \cdot P_j) = 2^n$ by assumption, the Cauchy-Schwarz inequality is tight, which only happens if $U=AP_jA^{\dagger}$ and $V=P_j$ are linearly dependent.
Thus, there exists a complex number $\lambda$ such that $AP_jA^{\dagger} = \lambda P_j$.
Substituting this expression in $\Tr(AP_jA^{\dagger} \cdot P_j)$ yields $%
\Tr(\lambda P_j \cdot P_j) = \lambda \cdot \Tr(I^{\otimes n}) = \lambda 2^n$, hence $\lambda = 1$ and $A P_j A^{\dagger} = P_j$.
\end{proof}

Combining \autoref{th:wmc} and \autoref{prop:encoding} with \autoref{thm:main-theorem} yields the following corollary.
\autoref{prop:wmc} directly leads to \autoref{alg:eq} which reduces equivalence checking to WMC.

\begin{corollary}\label{prop:wmc}
Let $A$ be an $n$-qubit circuits.
Given a Pauli string $\mathbf P ^0$
and $\mathbf P ^m \defn A \mathbf P ^0 A^{\dagger}$,
which is encoded by $F_{\mathbf P_m}$
with according weight function $W$,
we have $A \equiv \unit$ if and only if $MC_W(F_{\mathbf{P}_m} \wedge F_{P_j}) = 1$
 for all $P_j \in \set{X_j, Z_j \mid j\in [n]}$.
\end{corollary}

\begin{algorithm}
\caption{Quantum circuit equivalence checking algorithm based on WMC.
 Given an $n$-qubit circuit $A=(G_0, G_1, \dots, G_{m-1})$, the algorithm decides whether $A$ is equivalent to the identity circuit.}
\label{alg:eq}
\begin{algorithmic}[1]
\For{$P \in \{X, Z\}$}
\For{$j \in \{1, 2, \dots, n\}$}
\State $\mathbf P^0 \gets +P_j$ 
\State $F_M \gets F_{\mathbf P^0}(\vec v^0)$ 
\For{$k$ ranging from $0$ to $m-1$}
\State $F_M \gets F_M \wedge F_{G_k}(\vec{v}^k)$
\EndFor
\If{WMC($F_M \wedge F_{P_j}(\vec v^m)) \neq 1$ \label{line:equiv-check-wmc}}
\Comment Following \autoref{prop:wmc}
\State \textbf{return} `not equivalent'
\EndIf
\EndFor
\EndFor
\vspace{1mm}
\State \textbf{return} `equivalent'
\end{algorithmic}
\end{algorithm}

\begin{example}
Consider $A = V^{\dagger} U$ where circuit $U = (T, T)$ and $V = (S)$ as in \autoref{ex:eqpauli}.
We show how to reduce the equivalence check $A\stackrel{?}{\equiv} \unit$ to weighted model counting.
First, we consider boolean constraint $F_1 \defn F_{AXA^\dagger}\wedge F_{X}$ encoding the check $A X A^{\dagger} \stackrel{?}{=} X$.
\newline
\scalebox{0.9}{    \begin{minipage}{0.3\textwidth}
    \scalebox{.8}{
    \begin{quantikz}[row sep={0.1cm}]
    \slice{$P^0$} &
    \gate{T}\slice{$P^1$} \gategroup[1,steps=2,style={dashed,rounded corners,fill=blue!20, inner ysep=0pt, inner xsep=0pt},
    			         background,label style={label position=below,anchor=north,yshift=-0.2cm}]{{$U$}}
      & \gate{T}\slice{$P^2$} &
      \gate{S^\dagger}\slice{$P^3$} \gategroup[1,steps=1,style={dashed,rounded corners,fill=blue!20, inner ysep=0pt, inner xsep=0pt}, background,label style={label position=below,anchor=north,yshift=-0.2cm}]{{$V^{\dagger}$}} &
    \end{quantikz}
    }
    \end{minipage}}
    \begin{minipage}{0.6\textwidth}
        \[F_1 =~ \hspace{-1ex}\underbrace{x^0_1\no{z^0_1}\no{r^0}\land T^0_1 \land T^1_1 \land S^{\dag,2}_1}_{F_{AXA^\dagger}} \land \ \hspace{-1ex} \underbrace{x^3_1\no{z^3_1}}_{F_X}\]
    \end{minipage}
\newline
The satisfying assignments of $F_1$ are
\begin{align*}
  SAT{(F_1)} = \{
  &  x^1_1 \no{z^1_1} \no{r^1_1} \ 
  x^2_1 {z^2_1} \no{r^2_1} \ 
  x^3_1 \no{z^3_1} \no{r^3_1} \ 
  u^0 u^1, \
  x^1_1 \no{z^1_1} \no{r^1_1} \ 
  x^2_1 {z^2_1} \no{r^2_1} \ 
  x^3_1 \no{z^3_1} \no{r^3_1} \ 
  u^0 u^1\}.
\end{align*}
so the weighted model count is
$MC_W(F_1) = \sum_{\sigma\in SAT(F_1)}W(\sigma(r^3_1))W(\sigma(u^0))W(\sigma(u^1)) = \tfrac{1}{\sqrt{2}}\cdot\tfrac{1}{\sqrt{2}} + \tfrac{1}{\sqrt{2}}\cdot\tfrac{1}{\sqrt{2}} = 1$.

Now we turn to the check $AZA^{\dagger} \stackrel{?}{=} Z$, obtaining the formula
\[
    F_2 \defn \ \hspace{-1ex}\underbrace{\no{x^0_1}z^0_1\no{r^0}  \land T^0_1 \land T^1_1 \land S^{\dag,2}_1}_{F_{AZA^\dagger}} \land \hspace{-1ex} \ \underbrace{\no{x^3_1} z^3_1}_{F_Z}
\]
which has satisfying assignments $SAT(F_2) = \{\no{x^0_1} {z^0_1} \no{r^0_1}\no{x^1_1} {z^1_1} \no{r^1_1} \no{x^2_1} {z^2_1} \no{r^2_1}\no{x^3_1} {z^3_1} \no{r^3_1} \no{u^0} \no{u^1}\}$,
and $MC_W(F_2)=W(\no{r^3_1})W(\no{u^0})W(\no{u^1}) = 1$.
Since both weighted model counts evaluate to $1$, we conclude that $A\equiv \unit$.
\qed
\end{example}

\section{Implementation: the ECMC tool}

We implemented our method in an open-source tool called \ecmc, available at 
\url{https://github.com/System-Verification-Lab/Quokka-Sharp}.
\ecmc takes two quantum circuits in QASM format~\cite{cross2022openqasm} as input. 
It encodes these circuits to a sequence of $2n$ weighted conjunctive normal form (CNF) formulas
as explained in \autoref{sec3}, and then uses the weighted model counter GPMC~\cite{hashimoto2020gpmc} to solve these constraints in parallel with 16 processes,
terminating as soon as one returns a negative result.
We choose GPMC as it supports the negative weights in our encoding and
performs the best among solvers with that capability in the model counting competition 2023~\cite{modelcounting2023}.
To demonstrate the effectiveness of our method,
we conducted a set of broad experiments as discussed in the following.

We performed equivalence checking of quantum circuits
comparing our method against the state-of-the-art tool \qcec~\cite{advanced2021burgholzer}, which runs
different algorithms and heuristics based on ZX calculus and decision diagrams (shorted as DD) in portfolio with 16 parallel threads~\cite{xu2009satzilla2009}.
Similar to \ecmc,
\qcec also terminates earlier when one thread returns ``non-equivalent''.
Since the ZX-calculus based method is still incomplete for universal quantum circuits,
in the sense that it is only capable of proving equivalence,
we use this tool under two settings:
one is the default setting which uses DD and ZX calculus in portfolio;
the other is to exclusively enable DD~\cite{advanced2021burgholzer}.
We use two families of circuits:
(i) random Clifford+$T$ circuits, which mimic hard problems arising in quantum chemistry~\cite{wright2022chemistry} and quantum many-body physics~\cite{random2023fisher}; 
(ii) all benchmarks from the public benchmark suite MQT Bench~\cite{mqt2023quetschlich},
which includes many important quantum algorithms like QAOA, VQE, QNN, Grover, etc.
All experiments have been conducted on a 3.5 GHz M2 Machine with MacOS 13 and 16 GB RAM. 
We set the time limit to be 5 minutes (300 seconds) and include the time to read a QASM file, construct the weighted CNF and perform the model counting in all reported runtimes.

\begin{figure}[ht] \label{ fig7} 
  \begin{minipage}[b]{0.5\linewidth}
    \includegraphics[width=\linewidth]{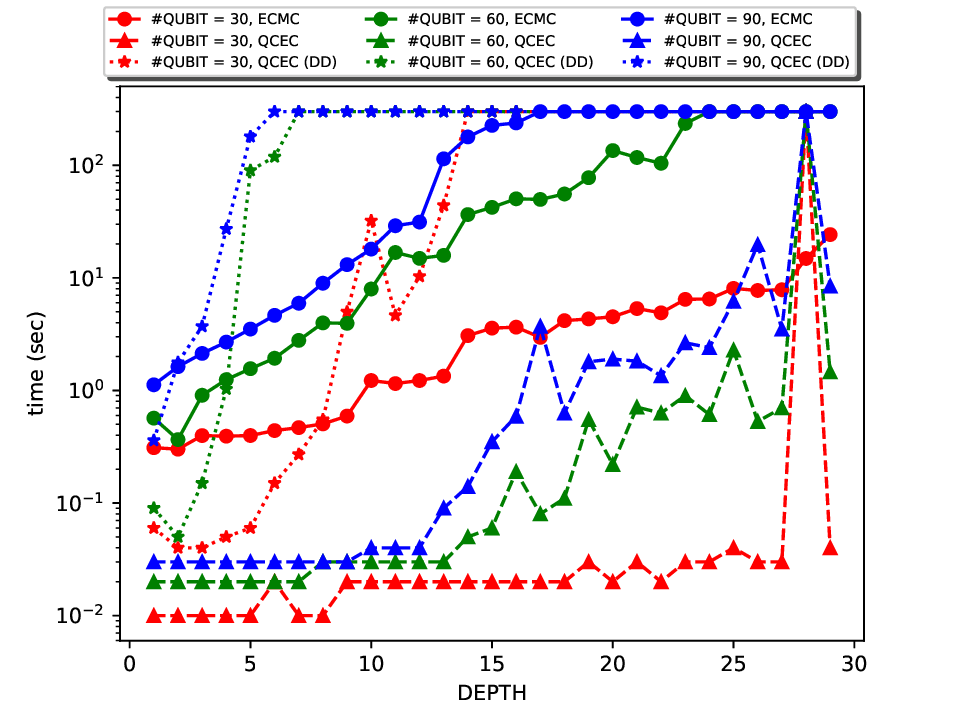} 
       \caption{Growing circuit depth} 
  \end{minipage} 
  \begin{minipage}[b]{0.5\linewidth}
    \includegraphics[width=\linewidth]{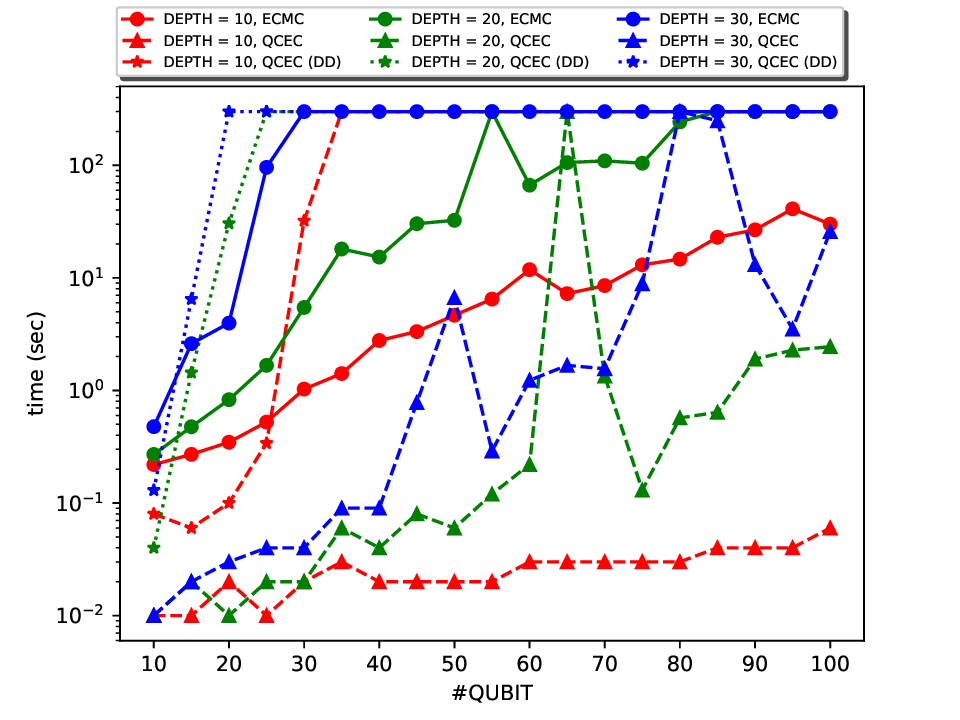} 
    \caption{Growing qubit counts} 
  \end{minipage} 
  \begin{minipage}[b]{0.5\linewidth}
    \includegraphics[width=\linewidth]{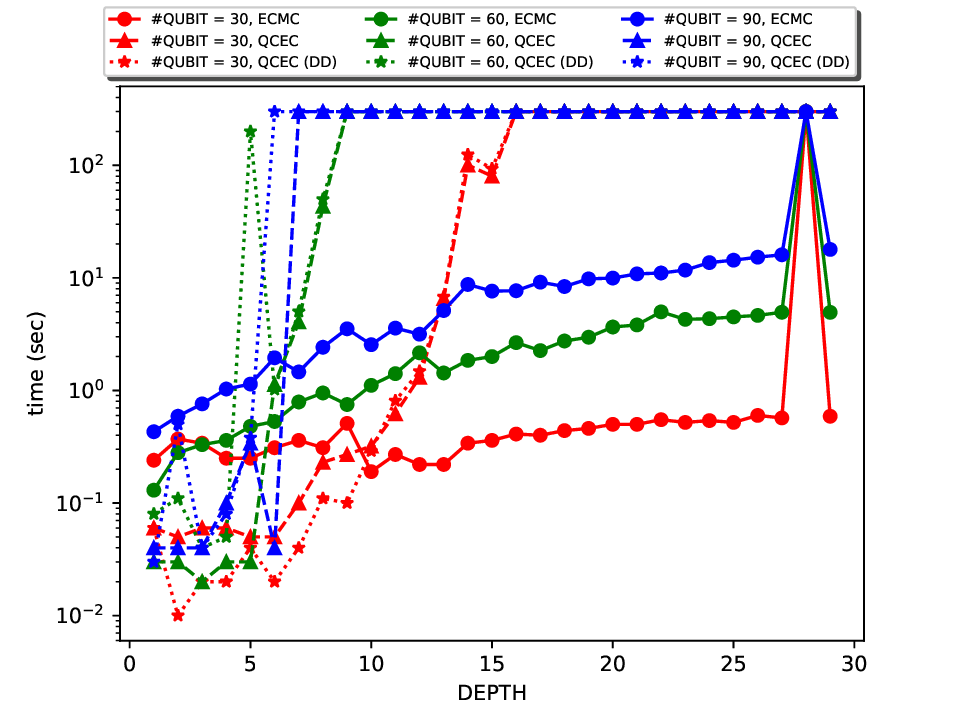} 
    \caption{Growing circuit depth} 
  \end{minipage}
  \hfill
  \begin{minipage}[b]{0.5\linewidth}
    \includegraphics[width=\linewidth]{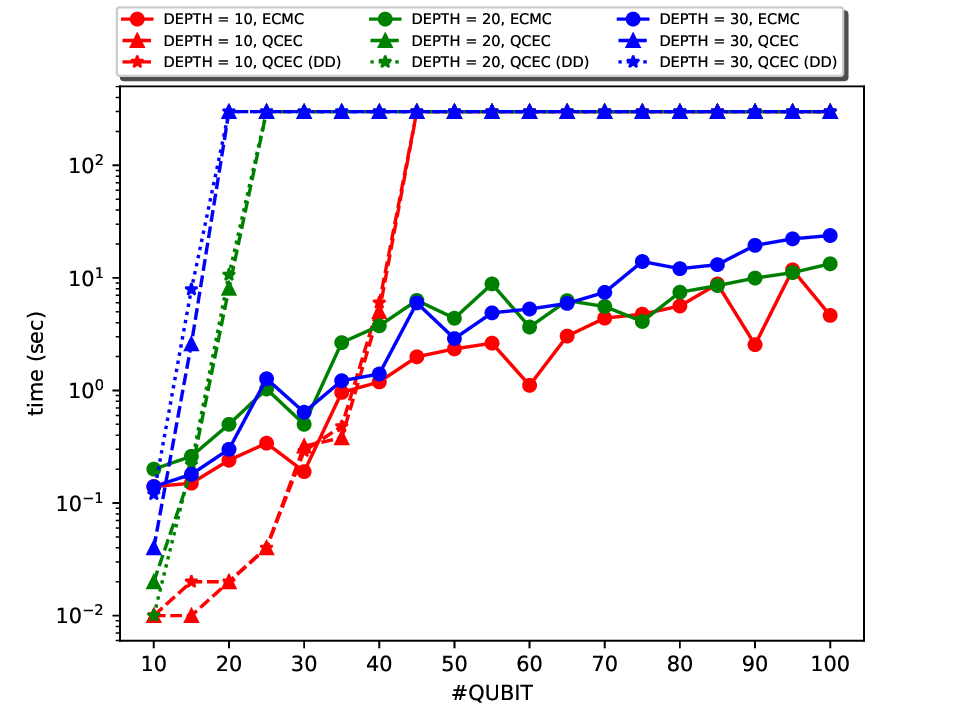} 
    \caption{Growing qubit counts} 
  \end{minipage} 
    \caption{Equivalence check of typical random Clifford+T circuits against their optimized circuits (equivalent cases, Fig 1 \& Fig 2) and optimized circuits with one random gate missing (non-equivalent cases, Fig 3 \& Fig 4). (Both vertical axes are on a logarithmic scale.)
  }
  \label{fig:scale}
\end{figure}

\subsection{Results}
First, 
to show the scalability of both methods on checking equivalence,
we consider random circuits that resemble typical oracle implementations --- random quantum circuits with varying qubits and depths, which comprise the $CX$, $H$, $S$, and $T$ gates with appearing ratio 10\%, 35\%,  35\%,  20\%~\cite{peham2022equivalence}.
We use a ZX-calculus tool PyZX~\cite{PyZXLS} to generate optimized circuits, to construct equivalent, yet very different, counterparts.
To construct non-equivalent instances,
we inject an error by removing one random gate from the corresponding optimized circuits.
So by construction, we know the correct answer for all equivalence checking instances in advance.
The resulting runtimes can be seen in \autoref{fig:scale}.

In addition to random circuits, 
to test structural quantum circuits,
we empirically evaluated our method on the MQTBench benchmark set~\cite{mqt2023quetschlich}.
We also generate the optimized circuits of the circuits from MQT-bench using PyZX~\cite{PyZXLS}.
To generate non-equivalent instances,
three kinds of errors are injected into the optimized circuits:
one with an random gate removed,
one where a random CNOT gate is flipped, switching control and target qubits,
and one where the phase of the angle of a random rotation gate is shifted.
For the last error, 
since many optimizations on rotation gates involve phase shifts in the rotation angles,
we consider two sizes of phase shift: one with the angle of a random rotation gate added by $10^{-4}$,
one with the angle added by $10^{-7}$.
We note that this experimental setup is stronger than the one used in~\cite{peham2022equivalence}, where only two errors are considered: bit flip and phase shift without giving the shifting scale.

We present a representative subset of equivalence checking results in \autoref{tab:algeq}.
The first three columns list the number of qubits $n$, the number of gates $|G|$ in original circuits and the number of gates $|G'|$ in optimized circuits.
Then we give the runtime of the weighted model counting tool \ecmc, the decision diagram-based \qcec (DD) and the default setting of \qcec respectively.
For the non-equivalent cases,
we show the flipped-CNOT and one-gate-missing error in \autoref{tab:fpgm}.
\begin{table}[b!]
  \setlength{\tabcolsep}{6pt} %
     \caption{Results of verifying equivalence of circuits from MQT bench against optimized circuits. 
     For cases within time limit,
     we give runtime (sec),
     while $>300$ represents a timeout (5 min) and \noinfo
     \ means that the result was `unknown'.
     }
     \label{tab:algeq}
    \centering
    \scriptsize
    \begin{tabular}{ c || r r r || c | c | c}
      \hline
{Algorithm} & {n} & {$|G|$} & {$G'$}  & {\ecmc} &  {\qcec (DD)}   & \qcec\\
      \hline
\multirowcell{3}{graphstate}  &  16  &  160  &  32  &  0.41  &  0.11 & 0.01 \\
  &  32  &  320  &  64  &  1.67  &  0.1 & 0.01 \\
  &  64  &  640  &  128  &  8.11  &  24.37 & 0.02\\ \hline
\multirowcell{3}{grover \\ (noancilla)}   &  5  &  499   &  629  &  \timeout  &  0.12 & 0.04\\
  &  6  &  1568  &  1870  &  \timeout  &  6.04 & \noinfo \\ 
  &  7  &  3751  &  5783  &  \timeout  &  \timeout & 1.97 \\ \hline
\multirowcell{3}{qaoa}  
  &  7  &  133  &  117  &  0.48  &  0.02  & 0.01 \\ 
  &  9  &  171  &  296  &  1.44  &  0.11 & \noinfo \\ 
  &  11  &  209  &  359  &  1.56  &  0.28 & 0.01 \\ 
  \hline
    \multirowcell{3}{qnn} &  2  &  43  &  36  &  0.06  &  0.01 & 0.01 \\
  &  8  &  319  &  494  &  \timeout  &  0.24 &\noinfo\\ 
  &  16  &  1023  &  2002  &  \timeout  &  \timeout & \noinfo\\ \hline
 \multirowcell{3}{qft}  &  2  &  14  &  14  &  0.02  &  0.01 & 0.01\\ 
  &  8  &  176  &  228  &  36.56  &  0.07 & 0.02 \\ 
 &  16  &  672  &  814  &  \timeout  &  18.28 & \noinfo\\ \hline
 \multirowcell{3}{qpe \\ (inexact)}   &  16  &  712  &  848  &  \timeout  &  \timeout & \noinfo \\
 &  32  &  2712  &  3179  &  \timeout  &  \timeout &  \timeout \\
 &  64  &  10552  &  9695  &  \timeout  &  \timeout &  \timeout\\ \hline
\multirowcell{3}{vqe}  &  5  &  83  &  83  &  1.07  &  0.01 & 0.01\\ 
&  10  &  168  &  221  &  \timeout  &  0.04 & 0.02\\ 
&  15  &  253  &  349 &  \timeout  &  0.12 & 0.05\\ \hline
\multirowcell{3}{wstate}  &  16  &  271  &  242  &  1.74  &  23.68 & \noinfo\\ 
  &  32  &  559  &  498  &  9.52  &  \timeout & \noinfo\\ 
  &  64  &  1135  &  1010  &  70.51  &  \timeout &  \timeout\\ \hline
     \end{tabular}
  \end{table}

The first three columns are the same as \autoref{tab:algeq} and then the performance of all three tools on CNOT flipped error and one-gate-missing error respectively. 
Finally, \autoref{tab:shift} shows the performance of phase shift errors,
where Shift-$10^{-4}$ (resp. Shift-$10^{-4}$) denotes adding $10^{-4}$ (resp. $10^{-7}$) to the phase of a random rotation gate.
The complete results can be found in \hyperref[appendix]{Appendix~\ref*{appendix}}.

\begin{table}[b!]
  \setlength{\tabcolsep}{2pt} %
     \caption{Results of verifying non-equivalence of circuits from MQT bench against optimized circuits with flipped CNOT gate (Flipped) and one missing gate (1 Gate Missing). 
     For cases within time limit,
     we give runtime (sec),
     while $>300$ represents a timeout (5 min).
     }
     \label{tab:fpgm}
    \centering
    \scriptsize
    \begin{tabular}{ c || r r r | c | c | c | c | c | c }
      \hline
    \multirowcell{2}{Algorithm}& \multirowcell{2}{n} & \multirowcell{2}{$|G|$}  & \multirowcell{2}{$|G'|$} & \multicolumn{3}{c|}{Flipped}  &  \multicolumn{3}{c}{1 Gate Missing}  \\
       &  &  &  & {\ecmc} &  {\qcec} (DD) & \qcec & {\ecmc} &  {\qcec} (DD)  & \qcec \\
       \hline
\multirowcell{3}{grover \ (noancilla)}  &  5  &  499  &  629  &  \timeout  &  0.04 & 0.02 &  1.26  &  0.05 & 0.02 \\ 
  &  6  &  1568  &  1870  &  \timeout  &  0.14 & 0.06  &  \timeout  &  0.13 & 0.05\\
  &  7  &  3751  &  5783  &  \timeout  &  1.41 & 0.46  &  24.27  &  5.05 & 0.35 \\ \hline
\multirowcell{3}{qaoa}  &  7  &  133  &  117  &  0.36  &  0.03  & 0.01 & 0.39  &  0.03 & 0.01\\ 
  &  9  &  171  &  296  &  0.77  &  0.06 & 0.03 &  0.8  &  0.07 & 0.02\\ 
  &  11  &  209  &  359  &  3.32  &  0.53  & 0.23 &  2.28  &  0.87 & 0.09 \\ \hline
\multirowcell{3}{qft}  &  2  &  14  &  14  &  0.05  &  0.02 & 0.01  &  0.1  &  0.04 & 0.01 \\ 
  &  8  &  176  &  228  &  1.53  &  0.05 & 0.01 &  1.89  &  0.04 &0.01 \\ 
  &  16  &  672  &  814  &  2.7  &  12.47 & 5.47 &  6.68  &  75.73 & 1.33\\ \hline
\multirowcell{3}{qnn}   &  2  &  43  &  36  &  0.24  &  0.07 & 0.01 &  0.22  &  0.1 & 0.01\\ 
  &  8  &  319  &  494  &  \timeout  &  0.59 & 0.05  &  \timeout  &  0.61 & 0.05\\ 
  &  16  &  1023  &  2002  &  \timeout  &  \timeout &  97.22  &  \timeout  &  \timeout & 90.58\\ \hline
\multirowcell{3}{qpe\\(inexact)}  &  16  &  712  &  848  &  19.97  &  1.72 & 4.98 &  19.59  &  186.29 & 1.12 \\ 
  &  32  &  2712  &  3179  &  13.28  &  \timeout  & \timeout & 22.0  &  \timeout &  \timeout \\ 
  &  64  &  10552  &  9695  &  \timeout  &  \timeout  & \timeout & 75.46  &  \timeout &  \timeout \\ \hline
\multirowcell{3}{vqe}  &  5  &  83  &  83  &  0.81  &  0.05 & 0.01 &  0.37  &  0.22 & 0.01\\ 
  &  10  &  168  &  221  &  55.06  &  0.58  & 0.18 &  3.98  &  3.9 & 0.03\\ 
  &  15  &  253  &  349  &  4.08  &  0.94  & 81.03 &  5.09  &  \timeout & 0.05 \\ \hline
\multirowcell{3}{wstate}  &  16  &  271  &  242  &  6.47  &  0.41 & 0.03 &  1.46  &  0.37 & 0.02 \\ 
 &  32  &  559  &  498  &  13.65  &  2.0  & \timeout &  2.28  &  \timeout & 59.07 \\ 
  &  64  &  1135  &  1010  &  13.32  &  \timeout &  \timeout  &  6.48  &  \timeout & \timeout\\ \hline
\end{tabular}
  \end{table}

\subsection{Discussion}
For random circuits,
\autoref{fig:scale} shows that the runtime of \ecmc  exhibits a clear correlation with the size of the circuits,
while \qcec and \qcec(DD) are very fast for small size circuits,
for non-equivalent cases, both of them are less scalable and reach time limit much earlier than \ecmc,
For the equivalent cases,
\qcec benefits from ZX calculus and outperforms the other two methods.
We suspect that \qcec (DD) shows poor performance when solving random circuits because these circuits don't contain the structure found in quantum algorithms, which decision diagrams can typically exploit.

\begin{table}[t!]
    \caption{Results of verifying non-equivalence of circuits from MQT bench against optimized circuits with $10^{-4}$ size and $10^{-7}$ size phase shift in one random rotation gate. 
     For cases within time limit,
     we give runtime (sec),
     while $>300$ represents a timeout (5 min), ``wrong''  a wrong result
    and \noinfo \  that the results was `unknown'.
     }
    \label{tab:shift}
    \scriptsize
  \setlength{\tabcolsep}{1pt} %
    \centering
    \begin{tabular}{ c | r r r| c c c | c c c}
      \hline
    \multirowcell{2}{Algorithm}& \multirowcell{2}{n} & \multirowcell{2}{$|G|$} & \multirowcell{2}{$|G'|$}  & \multicolumn{3}{c|}{Shift-$10^{-4}$}  &  \multicolumn{3}{c}{Shift-$10^{-7}$} \\
       &  & &  &  {\ecmc} & {QCEC} (DD) & {QCEC}  &  {\ecmc} & {QCEC (DD)}  & {QCEC} \\
       \hline
\multirowcell{3}{groundstate} & 4 & 180 & 36 & 0.26 & \ERROR & \ERROR & \ERROR & \ERROR & \ERROR\\& 12 & 1212 & 164 & \timeout & \ERROR & \ERROR & \timeout & \ERROR & \ERROR\\ 
 & 14 & 1610 & 206 & \timeout & \ERROR & \ERROR & \timeout & \ERROR & \ERROR\\ \hline
\multirowcell{3}{qaoa} & 7 & 133 & 117 & 0.15  & 0.03 & \noinfo & 0.15  & \timeout & \noinfo \\
& 9 & 171 & 296 & 0.29  &{\ERROR} & \noinfo & 0.32  & \timeout & \noinfo\\
& 11 & 209 & 359 & 0.33  & 0.12 & 0.1 & 0.32  & \timeout & \ERROR\\
\hline
\multirowcell{3}{qft} & 2 & 14 & 14 & 0.02  & 0.01 & 0.01 & 0.02  & 0.01 & 0.01å \\
& 8 & 176 & 228 & 0.2  &{\ERROR} & \noinfo & 0.21  & \timeout &  \noinfo\\
& 16 & 672 & 814 & 0.79  &{\ERROR} & \noinfo & 0.92  & \timeout & \noinfo \\
\hline
\multirowcell{3}{qnn} & 2 & 43 & 36 & 0.04  &{\ERROR} & {\ERROR}  & 0.04  & {\ERROR} & {\ERROR}  \\
 & 8 & 319 & 494 & \timeout & 0.24 & {\ERROR} & \timeout & 56.55 & \noinfo  \\
 & 16 & 1023 & 2002 & \timeout & \timeout & \noinfo & \timeout& \timeout & \noinfo \\
\hline
\multirowcell{3}{qpeinexact} & 16 & 712 & 848 & 8.59  & \timeout & \noinfo & 11.9  & \timeout & \noinfo \\
& 32 & 2712 & 3179 & \timeout & \timeout & \timeout & \timeout & \timeout & \timeout \\
 & 64 & 10552 & 9695 & \timeout & \timeout & \timeout & \timeout & \timeout & \timeout \\
\hline
\multirowcell{3}{routing} & 2 & 43 & 29 & 0.06  &{\ERROR} & \ERROR & 0.05  & \timeout & \ERROR \\
 & 6 & 135 & 142 & 0.33  & 0.02 & 0.01 & 2.49  & 0.03 & 0.01 \\
& 12 & 273 & 409 & 144.3  & 0.05 & 0.03 & \timeout & 0.09 & 0.04 \\
\hline
\multirowcell{3}{wstate} & 16 & 271 & 242 & 0.33  & 12.85 & \noinfo & 0.23  & 11.67 &\noinfo \\
& 32 & 559 & 498 & 1.55  & \timeout & \noinfo & 1.28  & \timeout & \noinfo \\
& 64 & 1135 & 1010 & 5.4  & \timeout & \timeout & 5.24  & \timeout & \timeout \\
\hline
     \end{tabular}
  \end{table}

When considering structural quantum circuits,
the results vary between equivalent and non-equivalent instances.
For equivalent instances, \qcec (DD) significantly surpasses \ecmc on Grover, QFT and QNN,
primarily due to the decision diagram-based method's proficiency in handling circuits featuring 
repeated structures and oracles.
While for those circuits featuring a large number of rotation gates with various rotation angles, 
like graphstate and wstate, \ecmc demonstrates clear advantages. 
Moreover, 
the default \qcec is much faster than \qcec (DD) on all cases while it reports ``no information'' for many cases as ZX calculus method and decision diagram method give different answers.

For non-equivalent instances,
since \ecmc can terminate when a single out of $2n$ WMC calls returns a negative result,
it shows better performance than checking equivalence.
For example, in the case of QPE, where both tools face time constraints when checking equivalent instances,
\ecmc can efficiently demonstrate non-equivalence and resolve the majority of cases within the time limit, while both \qcec and \qcec (DD) still gets timeout in most instances.

In all instances, \ecmc outperforms both \qcec and \qcec (DD) on graph state and wstate, each featuring many rotation gates.
When dealing with rotation gates,
decision diagrams might suffer from numerical instability~\cite{9023381,peham2022equivalence,9023381},
as can be clearly observed in \autoref{tab:shift} for the instances with errors in the phase shift,
where both \qcec and \qcec (DD) get wrong results for many benchmarks.
In contrast, the WMC approach ---also numerical in nature--- iteratively computes a sum of products, 
which we think avoids numerical instability.  \autoref{tab:shift} also demonstrates this point as \ecmc yields the correct answer for most benchmarks with $10^{-4}$ and $10^{-7}$-size error.
In contrast, the default \qcec gives no answer for a large amount of cases.

\section{Related Work}
\vspace{-1em}

Bauer et al. \cite{Bauer2023symQV} verify quantum programs  by encoding the verification problem in SMT, using an undecidable theory of nonlinear real arithmetic with trigonometric expression. An SMT theory for quantum computing was proposed in~\cite{chen2023theory}.
Berent et al. \cite{berent2022towards} realize Clifford circuit simulator and equivalence checker based on a SAT encoding. The equivalence checker was superseded by the deterministic polynomial-time algorithm proposed and implemented in~\cite{thanos2023fast}.
Using weighted model counting, universal quantum circuit simulation is realized in~\cite{qcmc}, which we extend by providing encodings for the $CZ$ and Toffoli gates and which we apply to circuit equivalence checking according to the approach of~\cite{thanos2023fast}.
SAT solvers have proven successful in quantum compilation~\cite{thanos2024automated}, e.g., for reversible simulation of circuits \cite{robert2011atpg} and optimizing space of quantum circuits \cite{meuli2018sat,quist2023optimizing}.

The ZX calculus~\cite{coecke2011interactingZXAlgebra} offers a diagrammatic approach to manipulate and analyze quantum circuits. A circuit is almost trivially expressible as a diagram, but the diagram language is more powerful and circuit extraction is consequently \#\P-complete~\cite{hardnessExtraction}. It has proven enormously successful in applications from equivalence checking~\cite{peham2022equivalence,Wille_ZX}, to circuit optimization~\cite{PyZXLS} and simulation~\cite{kissinger_simulating_2022}.

Decision diagrams~\cite{qmdd} (DDs) have been applied to simulated quantum circuits and checking their equivalence~\cite{advanced2021burgholzer} and synthesis~\cite{zulehner2017improving}.
To reduce circuits, the work in~\cite{jimenezpastor2024forward} uses bisimulation. In some cases this approach reduces simulation time compared to DDs~\cite{jimenezpastor2024forward}.

\section{Conclusions}
\label{sec:conclusion}

We have shown how weighted model counting can be used for circuit equivalence checking.
To achieve this, we considered quantum states in the Pauli basis, which allows for an efficient reduction of the equivalence checking problem to weighted model counting. We extended a linear-length encoding with the three-qubit Toffoli gate, so that all common non-Clifford gates are supported (previously the $T$, phase shift and rotation gates were already supported).
Given two $n$-qubit quantum circuits, their equivalence (up to global phase) can be decided by $2n$ calls to a weighted model counter, each with an encoding that is linear in the circuit size. Our open source implementation demonstrates that this technique is competitive to state-of-the-art methods based on a combination of decision diagrams and ZX calculus. This result demonstrates the strength of classical reasoning tools can transfer to the realm of quantum computing, despite the general `quantum'-hardness of these problems.

In future work, we plan to extract diagnostics for non-equivalent circuits from the satisfying assignments of the model counter.

\bibliographystyle{plain}
\bibliography{lit} 

\newpage
\appendix
\section{Additional Experimental Results}
\label{appendix}

\autoref{tab:full1}, \autoref{tab:full2} and \autoref{tab:full3} provide additional experimental results.
For some benchmarks,
the optimized circuits do not contain CNOT or rotation gate,
so we do not include them in the table for non-equivalent instances.

\begin{table}[ht]
  \setlength{\tabcolsep}{6pt} %
   \renewcommand{\arraystretch}{.9} %
     \caption{Results of verifying circuits from MQT bench. 
     For cases within time limit,
     we give their running time (sec),
     where $>300$ represents a timeout (5 min).
     }
     \label{tab:full1}
    \centering
    \small
    \begin{tabular}{ c || r r r || c | c | c}
      \hline
{Algorithm} & {n} & {$|G|$} & {$G'$}  & {\ecmc} &  {\qcec (DD)} & \qcec  \\
      \hline
routing & 2 & 43 & 29 & 0.17 & 0.01 & 0.01\\ \hline
routing & 6 & 135 & 142 & 161.17 & 0.02 & 0.01\\ \hline
routing & 12 & 273 & 409 & \timeout & 0.2 & 0.04\\ \hline
ae & 16 & 830 & 802 & \timeout & \timeout & \noinfo\\ \hline
ae & 32 & 2950 & 2862 & \timeout & \timeout & \timeout\\ \hline
ae & 64 & 11030 & 7728 & \timeout & \timeout & \timeout\\ \hline
dj & 16 & 127 & 67 & 0.16 & 0.03 & 0.01\\ \hline
dj & 32 & 249 & 129 & 0.51 & 0.03 & 0.01\\ \hline
dj & 64 & 507 & 259 & 2.09 & 0.08 & 0.03\\ \hline
ghz & 16 & 18 & 46 & 0.09 & 0.02 & 0.01\\ \hline
ghz & 32 & 34 & 94 & 0.24 & 0.03 & 0.01\\ \hline
ghz & 64 & 66 & 190 & 0.82 & 0.04 & 0.02\\ \hline
groundstate & 14 & 1610 & 206 & \timeout & 0.26 & 0.26\\ \hline
groundstate & 12 & 1212 & 164 & \timeout & 0.08 & 0.06\\ \hline
groundstate & 4 & 180 & 36 & 0.93 & 0.01 & 0.01\\ \hline
grover-v-chain & 5 & 529 & 632 & 250.88 & 0.1 & 0.05\\ \hline
grover-v-chain & 7 & 1224 & 1627 & \timeout & 39.6 & 0.18\\ \hline
grover-v-chain & 9 & 3187 & 4815 & \timeout & \timeout & 1.52\\ \hline
portfolioqaoa & 5 & 195 & 236 & 53.32 & 0.05 & \noinfo\\ \hline
portfolioqaoa & 6 & 261 & 356 & \timeout & 1.52 & \noinfo\\ \hline
portfolioqaoa & 7 & 336 & 481 & \timeout & 2.12 & \noinfo\\ \hline
portfoliovqe & 5 & 310 & 131 & 71.64 & 0.03 & 0.01\\ \hline
portfoliovqe & 6 & 435 & 151 & \timeout & 0.12 & \noinfo\\ \hline
portfoliovqe & 7 & 581 & 218 & \timeout & 0.05 & \noinfo\\ \hline
pricingcall & 5 & 240 & 166 & 1.1 & 0.02 & \noinfo\\ \hline
pricingcall & 7 & 422 & 277 & 7.45 & 0.05 & \noinfo\\ \hline
pricingcall & 9 & 624 & 396 & 46.69 & 0.21 & \noinfo\\ \hline
pricingput & 5 & 240 & 192 & 0.71 & 0.02 & \noinfo\\ \hline
pricingput & 7 & 432 & 297 & 12.94 & 0.06 & \noinfo\\ \hline
pricingput & 9 & 654 & 428 & 95.26 & 0.29 & \noinfo\\ \hline
qftentangled & 16 & 690 & 853 & \timeout & 46.27 & 0.11\\ \hline
qftentangled & 32 & 2658 & 3068 & \timeout & \timeout & \timeout\\ \hline
qftentangled & 64 & 10434 & 9387 & \timeout & \timeout & \timeout\\ \hline
qwalk-noancilla & 6 & 2457 & 3003 & \timeout & 0.39 & 0.36\\ \hline
qwalk-noancilla & 7 & 4761 & 6231 & \timeout & 1.28 & 1.31\\ \hline
qwalk-noancilla & 8 & 9369 & 12486 & \timeout & \timeout & \noinfo\\ \hline
qwalk-v-chain & 5 & 417 & 398 & 64.0 & 0.04 & 0.02\\ \hline
qwalk-v-chain & 7 & 849 & 840 & \timeout & 0.17 & 0.15\\ \hline
qwalk-v-chain & 9 & 1437 & 1500 & \timeout & 1.67 & 0.61\\ \hline
realamprandom & 16 & 680 & 679 & \timeout & 241.57 & \noinfo\\ \hline
realamprandom & 32 & 2128 & 2215 & \timeout & \timeout & \timeout\\ \hline
realamprandom & 64 & 7328 & 7411 & \timeout & \timeout & \timeout\\ \hline
     \end{tabular}

  \end{table}

  \begin{table}[!ht]
  \setlength{\tabcolsep}{1pt} %
     \caption{Results of verifying non-equivalence of circuits from MQT bench against optimized circuits with flipped CNOT gate (Flipped) and one missing gate (1 Gate missing). 
     For cases within time limit,
     we give their running time (sec),
     where $>300$ represents a timeout (5 min).
     }
     \label{tab:full2}
    \centering
    \scriptsize
    \begin{tabular}{ c || r r r | c | c | c | c | c | c }
      \hline
  {Algorithm}& {n} & {$|G|$}  & {$|G'|$} & \multicolumn{3}{c|}{Flipped}  &  \multicolumn{3}{c}{1 Gate missing}  \\
         &  &  &  & {\ecmc} &  {\qcec} (DD) & \qcec & {\ecmc} &  {\qcec} (DD)  & \qcec \\
     \hline
ae & 16 & 830 & 802 & 84.56 & 137.32 & 169.62 & 3.9 & 140.86 & 169.62\\ \hline
ae & 32 & 2950 & 2862 & \timeout & \timeout & \timeout & \timeout & \timeout & \timeout\\ \hline
ae & 64 & 11030 & 7728 & \timeout & \timeout & \timeout & \timeout & \timeout & \timeout\\ \hline
dj & 16 & 127 & 67 & 0.08 & 0.02 & 0.01 & 0.15 & 0.03 & 0.01\\ \hline
dj & 32 & 249 & 129 & 0.26 & 0.03 & 0.02 & 0.27 & 0.03 & 0.02\\ \hline
dj & 64 & 507 & 259 & 0.28 & 0.06 & 0.04 & 1.32 & 0.06 & 0.04\\ \hline
groundstate & 14 & 1610 & 206 & \timeout & 0.32 & 0.23 & \timeout & 0.87 & 0.23\\ \hline
groundstate & 12 & 1212 & 164 & \timeout & 0.1 & 0.07 & \timeout & 0.08 & 0.07\\ \hline
grover-v-chain & 5 & 529 & 632 & 215.42 & 0.06 & 0.02 & 269.69 & 0.07 & 0.02\\ \hline
grover-v-chain & 7 & 1224 & 1627 & \timeout & 0.11 & 0.05 & \timeout & 0.06 & 0.05\\ \hline
grover-v-chain & 9 & 3187 & 4815 & \timeout & 0.69 & 1.25 & \timeout & 0.96 & 1.25\\ \hline
portfolioqaoa & 5 & 195 & 236 & 39.88 & 0.02 & 0.01 & 39.56 & 0.02 & 0.01\\ \hline
portfolioqaoa & 6 & 261 & 356 & \timeout & 0.12 & 0.02 & \timeout & 0.03 & 0.02\\ \hline
portfolioqaoa & 7 & 336 & 481 & \timeout & 0.12 & 0.04 & \timeout & 0.16 & 0.04\\ \hline
portfoliovqe & 5 & 310 & 131 & 50.17 & 0.02 & 0.01 & 55.71 & 0.02 & 0.01\\ \hline
portfoliovqe & 6 & 435 & 151 & \timeout & 0.03 & 0.02 & \timeout & 0.03 & 0.02\\ \hline
portfoliovqe & 7 & 581 & 218 & \timeout & 0.05 & 0.02 & \timeout & 0.04 & 0.02\\ \hline
pricingcall & 5 & 240 & 166 & 0.67 & 0.03 & 0.01 & 0.52 & 0.03 & 0.01\\ \hline
pricingcall & 7 & 422 & 277 & 5.91 & 0.04 & 0.02 & 5.81 & 0.04 & 0.02\\ \hline
pricingcall & 9 & 624 & 396 & 55.92 & 0.07 & 0.05 & 57.47 & 0.07 & 0.05\\ \hline
pricingput & 5 & 240 & 192 & 0.62 & 0.03 & 0.01 & 0.17 & 0.04 & 0.01\\ \hline
pricingput & 7 & 432 & 297 & 11.25 & 0.04 & 0.02 & 11.6 & 0.05 & 0.02\\ \hline
pricingput & 9 & 654 & 428 & 115.43 & 0.1 & 0.07 & 243.27 & 0.12 & 0.07\\ \hline
qftentangled & 16 & 690 & 853 & 2.17 & 69.98 & 16.48 & 0.79 & 16.13 & 16.48\\ \hline
qftentangled & 32 & 2658 & 3068 & 4.93 & \timeout & \timeout & 5.65 & \timeout & \timeout\\ \hline
qftentangled & 64 & 10434 & 9387 & \timeout & \timeout & \timeout & 32.62 & \timeout & \timeout\\ \hline
qwalk-noancilla & 6 & 2457 & 3003 & \timeout & 0.13 & 0.09 & \timeout & 0.13 & 0.09\\ \hline
qwalk-noancilla & 7 & 4761 & 6231 & \timeout & 0.24 & 0.27 & \timeout & 0.25 & 0.27\\ \hline
qwalk-noancilla & 8 & 9369 & 12486 & \timeout & 1.13 & 1.37 & \timeout & 0.87 & 1.37\\ \hline
qwalk-v-chain & 5 & 417 & 398 & 56.83 & 0.04 & 0.02 & 75.87 & 0.04 & 0.02\\ \hline
qwalk-v-chain & 7 & 849 & 840 & \timeout & 0.04 & 0.04 & \timeout & 0.04 & 0.04\\ \hline
qwalk-v-chain & 9 & 1437 & 1500 & \timeout & 0.12 & 0.13 & \timeout & 0.16 & 0.13\\ \hline
realamprandom & 16 & 680 & 679 & 6.05 & 126.81 & 153.11 & \timeout & 128.76 & 153.11\\ \hline
realamprandom & 32 & 2128 & 2215 & \timeout & \timeout & \timeout & 16.14 & \timeout & \timeout\\ \hline
realamprandom & 64 & 7328 & 7411 & \timeout & \timeout & \timeout & \timeout & \timeout & \timeout\\ \hline
tsp & 16 & 1005 & 623 & \timeout & 33.37 & 29.94 & \timeout & 32.48 & 29.94\\ \hline
tsp & 4 & 225 & 86 & 5.1 & 0.02 & 0.01 & 5.87 & 0.02 & 0.01\\ \hline
tsp & 9 & 550 & 315 & \timeout & 0.08 & 0.06 & \timeout & 0.18 & 0.06\\ \hline
twolocalrandom & 16 & 680 & 679 & 8.92 & 126.12 & 111.84 & \timeout & 122.51 & 111.84\\ \hline
twolocalrandom & 32 & 2128 & 2215 & \timeout & \timeout & \timeout & 16.62 & \timeout & \timeout\\ \hline

    \end{tabular}
  \end{table}

\begin{table}[t!]
    \caption{Results of verifying non-equivalence of circuits from MQT bench against optimized circuits with $10^{-4}$ size and $10^{-7}$ size phase shift in one random rotation gate. 
     For cases within time limit,
     we give their running time (sec),
     where $>300$ represents a timeout (5 min) and "wrong result" represents when the result is equivalence.
     }
    \label{tab:full3}
    \scriptsize
  \setlength{\tabcolsep}{1pt} %
    \centering
    \begin{tabular}{ c | r r r| c c c | c c c}
      \hline
    \multirowcell{2}{Algorithm}& \multirowcell{2}{n} & \multirowcell{2}{$|G|$} & \multirowcell{2}{$|G'|$}  & \multicolumn{3}{c|}{Shift-$10^{-4}$}  &  \multicolumn{3}{c}{Shift-$10^{-7}$} \\
       &  & &  &  {\ecmc} & {QCEC} (DD) & {QCEC}  &  {\ecmc} & {QCEC (DD)}  & {QCEC} \\
       \hline
dj & 16 & 127 & 67 & 0.17 & 0.03 & \noinfo & \ERROR & \ERROR & 0.01\\ \hline
dj & 32 & 249 & 129 & 0.77 & \ERROR & \noinfo & \ERROR & \ERROR & \noinfo\\ \hline
dj & 64 & 507 & 259 & 1.38 & 0.09 & \noinfo & \ERROR & \ERROR & \noinfo\\ \hline
groundstate & 14 & 1610 & 206 & \timeout & \ERROR & \ERROR & \timeout & \ERROR & \ERROR\\ \hline
groundstate & 12 & 1212 & 164 & \timeout & \ERROR & \ERROR & \timeout & \ERROR & \ERROR\\ \hline
groundstate & 4 & 180 & 36 & 0.26 & \ERROR & \ERROR & \ERROR & \ERROR & \ERROR\\ \hline
grover-v-chain & 5 & 529 & 632 & 276.69 & \ERROR & \ERROR & 285.93 & \ERROR & \noinfo\\ \hline
grover-v-chain & 7 & 1224 & 1627 & \timeout & 82.38 & \noinfo & \timeout & \ERROR & \noinfo\\ \hline
grover-v-chain & 9 & 3187 & 4815 & \timeout & \timeout & \noinfo & \timeout & \timeout & \noinfo\\ \hline
portfolioqaoa & 5 & 195 & 236 & 39.22 & 0.12 & \noinfo & \ERROR & \ERROR & \noinfo\\ \hline
portfolioqaoa & 6 & 261 & 356 & \timeout & 2.16 & \noinfo & \timeout & \ERROR & \noinfo\\ \hline
portfolioqaoa & 7 & 336 & 481 & \timeout & 8.75 & \noinfo & \timeout & \ERROR & \noinfo\\ \hline
portfoliovqe & 5 & 310 & 131 & 45.57 & \ERROR & \noinfo & \ERROR & \ERROR & \noinfo\\ \hline
portfoliovqe & 6 & 435 & 151 & \timeout & 0.07 & \noinfo & \timeout & \ERROR & \noinfo\\ \hline
portfoliovqe & 7 & 581 & 218 & \timeout & \ERROR & \noinfo & \timeout & \ERROR & \noinfo\\ \hline
pricingcall & 5 & 240 & 166 & 0.63 & \ERROR & \ERROR & 0.9 & \ERROR & \noinfo\\ \hline
pricingcall & 7 & 422 & 277 & 7.46 & 0.06 & 0.04 & 7.7 & \ERROR & \noinfo\\ \hline
pricingcall & 9 & 624 & 396 & 51.71 & 0.22 & \noinfo & 52.05 & \ERROR & \noinfo\\ \hline
pricingput & 5 & 240 & 192 & 0.66 & \ERROR & \noinfo & 0.65 & \ERROR & \noinfo\\ \hline
pricingput & 7 & 432 & 297 & 9.9 & \ERROR & \noinfo & 13.0 & \ERROR & \noinfo\\ \hline
pricingput & 9 & 654 & 428 & 69.41 & 0.31 & \noinfo & 95.91 & \ERROR & \noinfo\\ \hline
qftentangled & 16 & 690 & 853 & 0.78 & \ERROR & \ERROR & \timeout & \ERROR & \ERROR\\ \hline
qftentangled & 32 & 2658 & 3068 & 9.87 & \timeout & \timeout & \timeout & \timeout & \timeout\\ \hline
qftentangled & 64 & 10434 & 9387 & \timeout & \timeout & \timeout & \timeout & \timeout & \timeout\\ \hline
qwalk-noancilla & 6 & 2457 & 3003 & \timeout & \ERROR & \ERROR & \timeout & \ERROR & \ERROR\\ \hline
qwalk-noancilla & 7 & 4761 & 6231 & \timeout & 1.28 & 1.33 & \timeout & 7.95 & 4.64\\ \hline
qwalk-noancilla & 8 & 9369 & 12486 & \timeout & \timeout & \noinfo & \timeout & \timeout & \noinfo\\ \hline
qwalk-v-chain & 5 & 417 & 398 & 63.4 & \ERROR & \ERROR & 62.9 & \ERROR & \ERROR\\ \hline
qwalk-v-chain & 7 & 849 & 840 & \timeout & \ERROR & \noinfo & \timeout & \ERROR & \ERROR\\ \hline
qwalk-v-chain & 9 & 1437 & 1500 & \timeout & \timeout & \noinfo & \timeout & \ERROR & \noinfo\\ \hline
realamprandom & 16 & 680 & 679 & \timeout & \ERROR & \noinfo & \timeout & \ERROR & \noinfo\\ \hline
realamprandom & 32 & 2128 & 2215 & \timeout & \timeout & \timeout & \timeout & \timeout & \timeout\\ \hline
realamprandom & 64 & 7328 & 7411 & \timeout & \timeout & \timeout & \timeout & \timeout & \timeout\\ \hline
su2random & 16 & 744 & 378 & \timeout & \timeout & \noinfo & \timeout & \timeout & \noinfo\\ \hline
su2random & 32 & 2256 & 762 & \timeout & \timeout & \timeout & \timeout & \timeout & \timeout\\ \hline
su2random & 64 & 7584 & 1530 & \timeout & \timeout & \timeout & \timeout & \timeout & \timeout\\ \hline
tsp & 16 & 1005 & 623 & \timeout & \timeout & \noinfo & \timeout & \timeout & \noinfo\\ \hline
tsp & 4 & 225 & 86 & 4.13 & \ERROR & \noinfo & \ERROR & \ERROR & \noinfo\\ \hline
tsp & 9 & 550 & 315 & \timeout & 4.14 & \noinfo & \timeout & \ERROR & \noinfo\\ \hline
twolocalrandom & 16 & 680 & 679 & \timeout & \ERROR & \noinfo & \timeout & \ERROR & \noinfo\\ \hline
twolocalrandom & 32 & 2128 & 2215 & \timeout & \timeout & \timeout & \timeout & \timeout & \timeout\\ \hline
twolocalrandom & 64 & 7328 & 7411 & \timeout & \timeout & \timeout & \timeout & \timeout & \timeout\\ \hline
     \end{tabular}
  \end{table}

\section{Toffoli Encoding Table}
\label{app:toffoli}

\autoref{tab:toffoli-full} provides the brute forced results of the Toffoli gate's behavior in the Pauli basis.

  \begin{table}[ht!]
  \setlength{\tabcolsep}{10pt} %
   \renewcommand{\arraystretch}{.98} %
     \caption{Pauli operator lookup table for the Toffoli gate for in/output $\hP$~and~$\hQ$.}
     \label{tab:toffoli-full}
    \centering
    \scriptsize
    \begin{tabular}{ >{\centering\arraybackslash}p{1.8cm} | l }

\toprule
 $\hP\in \mathcal{\hat P}_3$ & $\hQ= \mathrm{Toffoli} \cdot \hP \cdot  \mathrm{Toffoli}^\dagger $ with $\hQ\in \frac 12  \sum_{i\in [4]} \mathcal{\hat P}_3 $ or $\hQ\in \mathcal{\hat P}_3$  \\
\midrule
  $ I \otimes I \otimes I $ & $   I \otimes I \otimes I $ \\
$ I \otimes I \otimes Z $ & $(   I \otimes I \otimes Z +   I \otimes Z \otimes Z +   Z \otimes I \otimes Z - Z \otimes Z \otimes Z ) / 2$ \\
$ I \otimes I \otimes X $ & $   I \otimes I \otimes X $ \\
$ I \otimes I \otimes Y $ & $(   I \otimes I \otimes Y +   I \otimes Z \otimes Y +   Z \otimes I \otimes Y - Z \otimes Z \otimes Y ) / 2$ \\
$ I \otimes Z \otimes I $ & $   I \otimes Z \otimes I $ \\
$ I \otimes Z \otimes Z $ & $(   I \otimes I \otimes Z +   I \otimes Z \otimes Z +   Z \otimes Z \otimes Z - Z \otimes I \otimes Z ) / 2$ \\
$ I \otimes Z \otimes X $ & $   I \otimes Z \otimes X $ \\
$ I \otimes Z \otimes Y $ & $(   I \otimes I \otimes Y +   I \otimes Z \otimes Y +   Z \otimes Z \otimes Y - Z \otimes I \otimes Y ) / 2$ \\
$ I \otimes X \otimes I $ & $(   I \otimes X \otimes I +   I \otimes X \otimes X +   Z \otimes X \otimes I - Z \otimes X \otimes X ) / 2$ \\
$ I \otimes X \otimes Z $ & $(   I \otimes X \otimes Z +   Z \otimes X \otimes Z +   Z \otimes Y \otimes Y - I \otimes Y \otimes Y ) / 2$ \\
$ I \otimes X \otimes X $ & $(   I \otimes X \otimes I +   I \otimes X \otimes X +   Z \otimes X \otimes X - Z \otimes X \otimes I ) / 2$ \\
$ I \otimes X \otimes Y $ & $(   I \otimes X \otimes Y +   I \otimes Y \otimes Z +   Z \otimes X \otimes Y - Z \otimes Y \otimes Z ) / 2$ \\
$ I \otimes Y \otimes I $ & $(   I \otimes Y \otimes I +   I \otimes Y \otimes X +   Z \otimes Y \otimes I - Z \otimes Y \otimes X ) / 2$ \\
$ I \otimes Y \otimes Z $ & $(   I \otimes X \otimes Y +   I \otimes Y \otimes Z +   Z \otimes Y \otimes Z - Z \otimes X \otimes Y ) / 2$ \\
$ I \otimes Y \otimes X $ & $(   I \otimes Y \otimes I +   I \otimes Y \otimes X +   Z \otimes Y \otimes X - Z \otimes Y \otimes I ) / 2$ \\
$ I \otimes Y \otimes Y $ & $(   I \otimes Y \otimes Y +   Z \otimes X \otimes Z +   Z \otimes Y \otimes Y - I \otimes X \otimes Z ) / 2$ \\
$ Z \otimes I \otimes I $ & $   Z \otimes I \otimes I $ \\
$ Z \otimes I \otimes Z $ & $(   I \otimes I \otimes Z +   Z \otimes I \otimes Z +   Z \otimes Z \otimes Z - I \otimes Z \otimes Z ) / 2$ \\
$ Z \otimes I \otimes X $ & $   Z \otimes I \otimes X $ \\
$ Z \otimes I \otimes Y $ & $(   I \otimes I \otimes Y +   Z \otimes I \otimes Y +   Z \otimes Z \otimes Y - I \otimes Z \otimes Y ) / 2$ \\
$ Z \otimes Z \otimes I $ & $   Z \otimes Z \otimes I $ \\
$ Z \otimes Z \otimes Z $ & $(   I \otimes Z \otimes Z +   Z \otimes I \otimes Z +   Z \otimes Z \otimes Z - I \otimes I \otimes Z ) / 2$ \\
$ Z \otimes Z \otimes X $ & $   Z \otimes Z \otimes X $ \\
$ Z \otimes Z \otimes Y $ & $(   I \otimes Z \otimes Y +   Z \otimes I \otimes Y +   Z \otimes Z \otimes Y - I \otimes I \otimes Y ) / 2$ \\
$ Z \otimes X \otimes I $ & $(   I \otimes X \otimes I +   Z \otimes X \otimes I +   Z \otimes X \otimes X - I \otimes X \otimes X ) / 2$ \\
$ Z \otimes X \otimes Z $ & $(   I \otimes X \otimes Z +   I \otimes Y \otimes Y +   Z \otimes X \otimes Z - Z \otimes Y \otimes Y ) / 2$ \\
$ Z \otimes X \otimes X $ & $(   I \otimes X \otimes X +   Z \otimes X \otimes I +   Z \otimes X \otimes X - I \otimes X \otimes I ) / 2$ \\
$ Z \otimes X \otimes Y $ & $(   I \otimes X \otimes Y +   Z \otimes X \otimes Y +   Z \otimes Y \otimes Z - I \otimes Y \otimes Z ) / 2$ \\
$ Z \otimes Y \otimes I $ & $(   I \otimes Y \otimes I +   Z \otimes Y \otimes I +   Z \otimes Y \otimes X - I \otimes Y \otimes X ) / 2$ \\
$ Z \otimes Y \otimes Z $ & $(   I \otimes Y \otimes Z +   Z \otimes X \otimes Y +   Z \otimes Y \otimes Z - I \otimes X \otimes Y ) / 2$ \\
$ Z \otimes Y \otimes X $ & $(   I \otimes Y \otimes X +   Z \otimes Y \otimes I +   Z \otimes Y \otimes X - I \otimes Y \otimes I ) / 2$ \\
$ Z \otimes Y \otimes Y $ & $(   I \otimes X \otimes Z +   I \otimes Y \otimes Y +   Z \otimes Y \otimes Y - Z \otimes X \otimes Z ) / 2$ \\
$ X \otimes I \otimes I $ & $(   X \otimes I \otimes I +   X \otimes I \otimes X +   X \otimes Z \otimes I - X \otimes Z \otimes X ) / 2$ \\
$ X \otimes I \otimes Z $ & $(   X \otimes I \otimes Z +   X \otimes Z \otimes Z +   Y \otimes Z \otimes Y - Y \otimes I \otimes Y ) / 2$ \\
$ X \otimes I \otimes X $ & $(   X \otimes I \otimes I +   X \otimes I \otimes X +   X \otimes Z \otimes X - X \otimes Z \otimes I ) / 2$ \\
$ X \otimes I \otimes Y $ & $(   X \otimes I \otimes Y +   X \otimes Z \otimes Y +   Y \otimes I \otimes Z - Y \otimes Z \otimes Z ) / 2$ \\
$ X \otimes Z \otimes I $ & $(   X \otimes I \otimes I +   X \otimes Z \otimes I +   X \otimes Z \otimes X - X \otimes I \otimes X ) / 2$ \\
$ X \otimes Z \otimes Z $ & $(   X \otimes I \otimes Z +   X \otimes Z \otimes Z +   Y \otimes I \otimes Y - Y \otimes Z \otimes Y ) / 2$ \\
$ X \otimes Z \otimes X $ & $(   X \otimes I \otimes X +   X \otimes Z \otimes I +   X \otimes Z \otimes X - X \otimes I \otimes I ) / 2$ \\
$ X \otimes Z \otimes Y $ & $(   X \otimes I \otimes Y +   X \otimes Z \otimes Y +   Y \otimes Z \otimes Z - Y \otimes I \otimes Z ) / 2$ \\
$ X \otimes X \otimes I $ & $(   X \otimes X \otimes I +   X \otimes X \otimes X +   Y \otimes Y \otimes I - Y \otimes Y \otimes X ) / 2$ \\
$ X \otimes X \otimes Z $ & $(   X \otimes X \otimes Z +   Y \otimes Y \otimes Z - X \otimes Y \otimes Y - Y \otimes X \otimes Y ) / 2$ \\
$ X \otimes X \otimes X $ & $(   X \otimes X \otimes I +   X \otimes X \otimes X +   Y \otimes Y \otimes X - Y \otimes Y \otimes I ) / 2$ \\
$ X \otimes X \otimes Y $ & $(   X \otimes X \otimes Y +   X \otimes Y \otimes Z +   Y \otimes X \otimes Z +   Y \otimes Y \otimes Y ) / 2$ \\
$ X \otimes Y \otimes I $ & $(   X \otimes Y \otimes I +   X \otimes Y \otimes X +   Y \otimes X \otimes X - Y \otimes X \otimes I ) / 2$ \\
$ X \otimes Y \otimes Z $ & $(   X \otimes X \otimes Y +   X \otimes Y \otimes Z - Y \otimes X \otimes Z - Y \otimes Y \otimes Y ) / 2$ \\
$ X \otimes Y \otimes X $ & $(   X \otimes Y \otimes I +   X \otimes Y \otimes X +   Y \otimes X \otimes I - Y \otimes X \otimes X ) / 2$ \\
$ X \otimes Y \otimes Y $ & $(   X \otimes Y \otimes Y +   Y \otimes Y \otimes Z - X \otimes X \otimes Z - Y \otimes X \otimes Y ) / 2$ \\
$ Y \otimes I \otimes I $ & $(   Y \otimes I \otimes I +   Y \otimes I \otimes X +   Y \otimes Z \otimes I - Y \otimes Z \otimes X ) / 2$ \\
$ Y \otimes I \otimes Z $ & $(   X \otimes I \otimes Y +   Y \otimes I \otimes Z +   Y \otimes Z \otimes Z - X \otimes Z \otimes Y ) / 2$ \\
$ Y \otimes I \otimes X $ & $(   Y \otimes I \otimes I +   Y \otimes I \otimes X +   Y \otimes Z \otimes X - Y \otimes Z \otimes I ) / 2$ \\
$ Y \otimes I \otimes Y $ & $(   X \otimes Z \otimes Z +   Y \otimes I \otimes Y +   Y \otimes Z \otimes Y - X \otimes I \otimes Z ) / 2$ \\
$ Y \otimes Z \otimes I $ & $(   Y \otimes I \otimes I +   Y \otimes Z \otimes I +   Y \otimes Z \otimes X - Y \otimes I \otimes X ) / 2$ \\
$ Y \otimes Z \otimes Z $ & $(   X \otimes Z \otimes Y +   Y \otimes I \otimes Z +   Y \otimes Z \otimes Z - X \otimes I \otimes Y ) / 2$ \\
$ Y \otimes Z \otimes X $ & $(   Y \otimes I \otimes X +   Y \otimes Z \otimes I +   Y \otimes Z \otimes X - Y \otimes I \otimes I ) / 2$ \\
$ Y \otimes Z \otimes Y $ & $(   X \otimes I \otimes Z +   Y \otimes I \otimes Y +   Y \otimes Z \otimes Y - X \otimes Z \otimes Z ) / 2$ \\
$ Y \otimes X \otimes I $ & $(   X \otimes Y \otimes X +   Y \otimes X \otimes I +   Y \otimes X \otimes X - X \otimes Y \otimes I ) / 2$ \\
$ Y \otimes X \otimes Z $ & $(   X \otimes X \otimes Y +   Y \otimes X \otimes Z - X \otimes Y \otimes Z - Y \otimes Y \otimes Y ) / 2$ \\
$ Y \otimes X \otimes X $ & $(   X \otimes Y \otimes I +   Y \otimes X \otimes I +   Y \otimes X \otimes X - X \otimes Y \otimes X ) / 2$ \\
$ Y \otimes X \otimes Y $ & $(   Y \otimes X \otimes Y +   Y \otimes Y \otimes Z - X \otimes X \otimes Z - X \otimes Y \otimes Y ) / 2$ \\
$ Y \otimes Y \otimes I $ & $(   X \otimes X \otimes I +   Y \otimes Y \otimes I +   Y \otimes Y \otimes X - X \otimes X \otimes X ) / 2$ \\
$ Y \otimes Y \otimes Z $ & $(   X \otimes X \otimes Z +   X \otimes Y \otimes Y +   Y \otimes X \otimes Y +   Y \otimes Y \otimes Z ) / 2$ \\
$ Y \otimes Y \otimes X $ & $(   X \otimes X \otimes X +   Y \otimes Y \otimes I +   Y \otimes Y \otimes X - X \otimes X \otimes I ) / 2$ \\
$ Y \otimes Y \otimes Y $ & $(   X \otimes X \otimes Y +   Y \otimes Y \otimes Y - X \otimes Y \otimes Z - Y \otimes X \otimes Z ) / 2$ \\

\bottomrule
    \end{tabular}
  \end{table}

\end{document}